\documentclass[conference, 12pt, onecolumn]{ieeeconf}

\IEEEoverridecommandlockouts
\usepackage[usenames]{color}
\usepackage{enumerate}
\usepackage{url}
\usepackage{subfigure}
\usepackage{amsfonts,mathrsfs}
\usepackage{amssymb,amsmath}
\usepackage{verbatim}
\usepackage{acronym}
\usepackage{mathtools}
\usepackage{cite}
\usepackage{graphicx}

\def\fskip#1{}

\newtheorem{theorem}{Theorem}

\newtheorem{definition}{Definition}
\newtheorem{example}{Example}

\newtheorem{remark}{Remark}

\renewcommand{\theenumi}{\arabic{enumi}}

\def\1{{\bf 1}}

\newcommand{\remove}[1]{}

\def\argmin{\mathop{\rm argmin}}

\begin{document}
\title{Duality and Stability in Complex Multiagent State-Dependent Network Dynamics}
\author{\authorblockN{S. Rasoul Etesami*}
\thanks{*Department of Industrial and Enterprise Systems Engineering and Coordinated Science Laboratory, University of Illinois at Urbana-Champaign, Urbana, IL, 61801. Email: etesami1@illinois.edu }
\thanks{This material is based upon work supported by the National Science Foundation under Grant No. EPCN-1944403.} 
}

\maketitle
\begin{abstract}
Despite significant progress on stability analysis of conventional multiagent networked systems with weakly coupled state-network dynamics, most of the existing results have shortcomings in addressing multiagent systems with highly coupled state-network dynamics. Motivated by numerous applications of such dynamics, in our previous work \cite{etesami2019simple}, we initiated a new direction for stability analysis of such systems that uses a sequential optimization framework. Building upon that, in this paper, we extend our results by providing another angle on multiagent network dynamics from a duality perspective, which allows us to view the network structure as dual variables of a constrained nonlinear program. Leveraging that idea, we show that the evolution of the coupled state-network multiagent dynamics can be viewed as iterates of a primal-dual algorithm for a static constrained optimization/saddle-point problem. This view bridges the Lyapunov stability of state-dependent network dynamics and frequently used optimization techniques such as block coordinated descent, mirror descent, the Newton method, and the subgradient method. As a result, we develop a systematic framework for analyzing the Lyapunov stability of state-dependent network dynamics using techniques from nonlinear optimization. Finally, we support our theoretical results through numerical simulations from social science.     
\end{abstract}


\section{Introduction}
Many of the current challenges in science and engineering are related to complex networks, and distributed
multiagent network systems are currently the focal point of many new applications. Such applications relate to the
growing popularity of social networks, the analysis of large network data sets, and the problems that arise
from interactions among agents in complex political, economic, and biological systems. These challenges may involve modeling of the interactions of agents in complex networks, the establishment of stability in the agents' interaction dynamics, and the design of efficient algorithms to obtain or approximate the equilibrium points. We can offer many motivating examples of relationships in political, social, and engineering applications that are governed by complex networks of heterogeneous agents. Agents may be strategic, or the networks can be dynamic in the sense that they can vary over time, depending on the agents' states or decisions. The following are a few examples.
  
\smallskip
\noindent
\emph{-- Network security}: A basic task in network security is that of providing a mechanism for securing the operation of a set of networked heterogeneous agents (e.g., service providers, computers, or data centers) despite external malicious attacks (Figure \ref{Fig:security}). One way of doing that is to incentivize the agents to invest in their security (e.g., by installing antivirus software) \cite{grossklags2008secure}. However, since the agents are interconnected, the compromise of one agent may affect its neighbors, and such a failure can cascade over the entire network. As a result, the decision made by each agent on how much to invest in its security level will indirectly affect all the others, and hence the connectivity structure of the network. Thus we face a highly dynamic network of heterogeneous agents in which the agents' states/decisions and the network structure are highly influenced by each other.

\smallskip
\noindent
\emph{-- Formation control}: A goal in formation control is to design a distributed protocol such that a set of agents (e.g., the aircraft in Figure \ref{Fig:aircraft}) collectively form a specific structure and eventually accomplish a task \cite{bullo2009distributed}. Agents may have different communication capabilities and can communicate only with those in their local neighborhoods. Consequently, depending on the agents' states (e.g., remaining power or relative positions), the communication network they share is subject to change. As a result, the agents' states and the communication network are highly coupled and dynamically evolve based on each other. 
 
\begin{figure}[!tbp]
  \centering
    \begin{minipage}[t]{0.32\textwidth}
    \includegraphics[width=\textwidth]{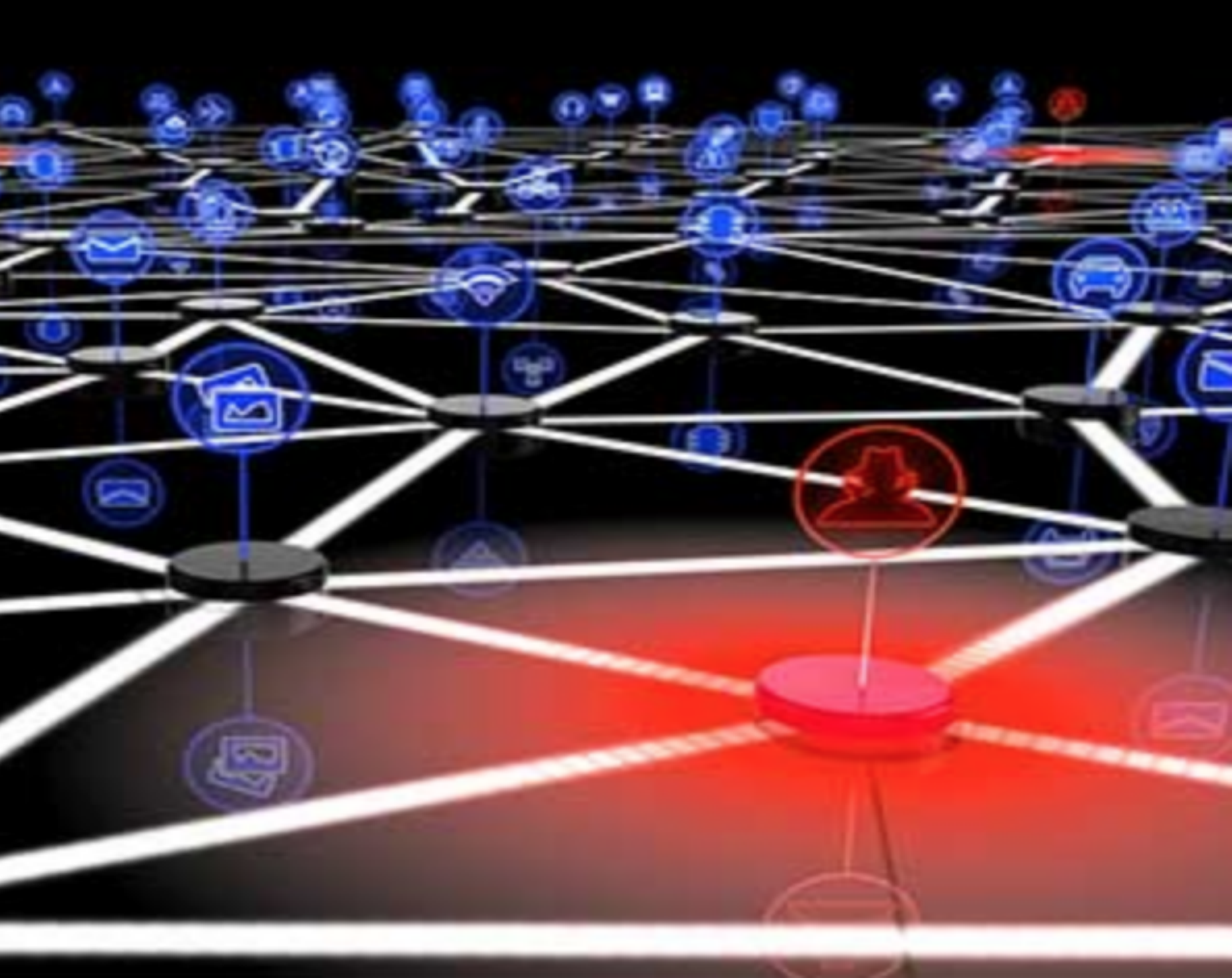}
    \vspace{-0.5cm} \caption{\footnotesize{Compromise of an agent changes the network structure and hence the security decisions of all other agents.}}\label{Fig:security}
  \end{minipage}
  \hfill 
  \begin{minipage}[t]{0.32\textwidth}
    \includegraphics[width=\textwidth]{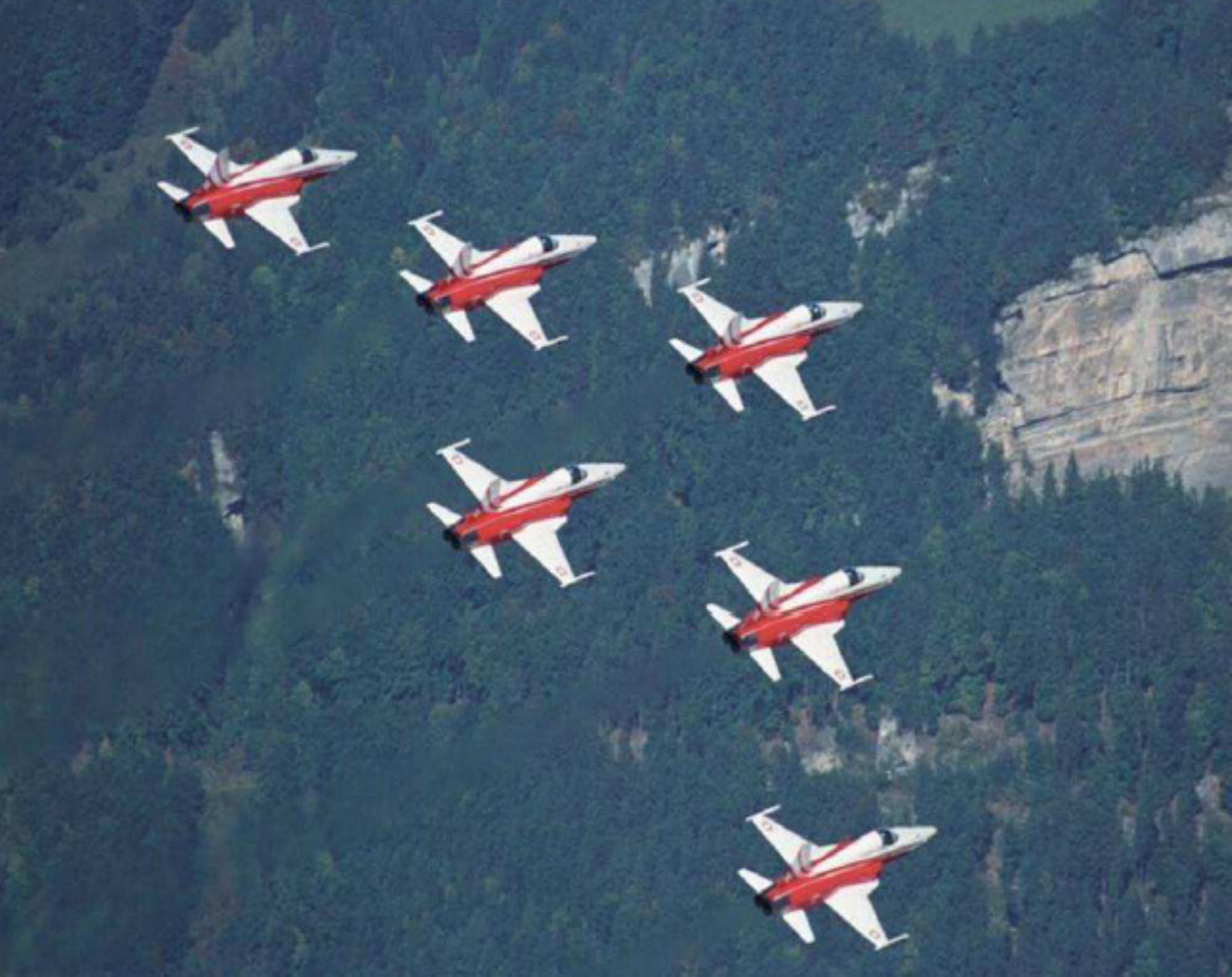}
    \vspace{-0.5cm} \caption{\footnotesize{Aircraft must keep a certain formation while the communication network among them is subject to change.}}\label{Fig:aircraft}
  \end{minipage}
  \hfill
  \begin{minipage}[t]{0.32\textwidth}
    \includegraphics[width=\textwidth]{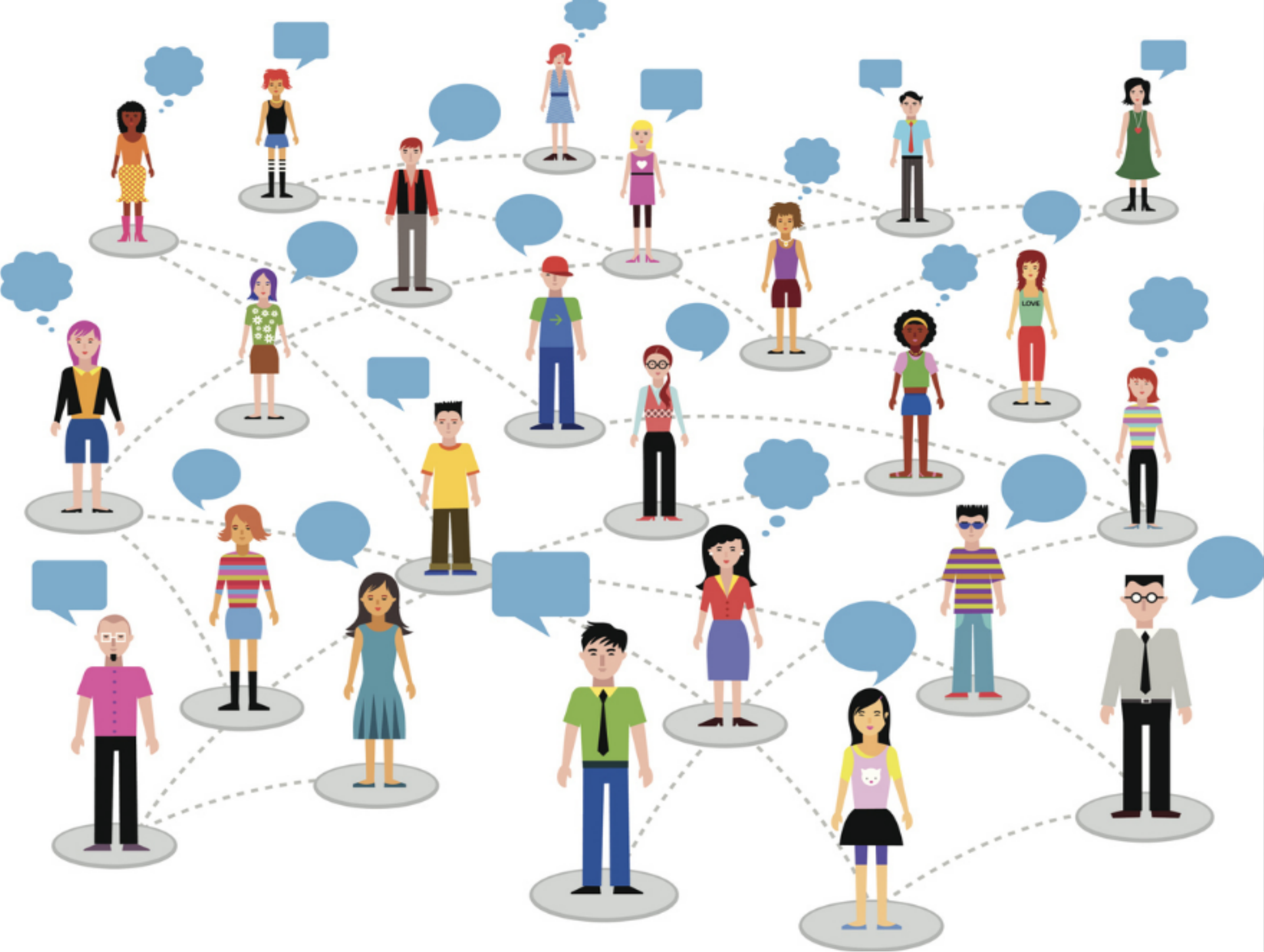}
    \vspace{-0.5cm} \caption{\footnotesize{The social network affects opinions, and that in turn creates new friendships and hence a new social network.}}\label{Fig:opinion}
  \end{minipage}
  \vspace{-0.4cm}
\end{figure}

\smallskip
\noindent 
\emph{-- Social networks}: In social networks, there are often apparent affinities among people based on heterogeneous political or cultural beliefs that define an interaction network among them. However, on specific issues, alliances form among people from different groups. Almost every congressional vote provides an example of this phenomenon, wherein some representatives break away from their respective parties to vote with the other party \cite{hegselmann2002opinion} (Figure \ref{Fig:opinion}).

\smallskip
\noindent 
\emph{-- Stability of smart grids}: In the emerging smart grid, a significant amount of energy stems from renewable sources, electric vehicles, and storage units, many of which may be owned by consumers rather than utility companies. That phenomenon is turning every grid component into a \emph{prosumer}: a joint producer and consumer of energy. Prosumers (agents) in the smart grid strategically interact with each other subject to power network constraints. In particular, depending on their states (e.g., energy consumption/production decisions), they may decide to buy/sell energy to different agents. Thus, a significant challenge is that of providing decentralized algorithms to stabilize the demand and response given that the structure of the agents' interactions is a function of their own and their neighbors' states/decisions.

Motivated by the above, and many other real applications such as robotics \cite{dimarogonas2008decentralized,zavlanos2007potential}, genomic \cite{rajapakse2011dynamics}, robustification \cite{cherukuri2017role}, and consensus seeking \cite{mesbahi2005state,awad2018time}, our objective in this paper is to provide a systematic approach for analyzing the stability and convergence of agents interacting over a rich dynamic network that may evolve or vary based on agents' states. To that end, we provide new connections between analysis of multiagent network systems and developed techniques in the mature field of nonlinear programming. Utilizing such connections, we show how Lyapunov stability of seemingly complex multiagent network dynamics can be analyzed using iterative optimization algorithms for finding a minimum or saddle-point of nonlinear functions.

\subsection{Related Work}

A general multiagent network problem involves a set of $[n]=\{1,2,\ldots,n\}$ agents (social individuals, grid prosumers, unmanned vehicles, etc.) At each time instance $k=0,1,2,\ldots$, there is an underlying network $\mathcal{G}_k=([n],\mathcal{E}_k)$ that determines the communication network shared by the agents. Here, $\mathcal{E}_k$ denotes the set of edges of the network at time $k$; the edges can be undirected or directed. The state of agent $i\in[n]$ at time $k$ is given by a vector $x^k_i$ and evolves based on $i$'s interaction with its neighbors. In particular, the overall state of the system at time $k+1$, denoted by $\boldsymbol{x}^{k+1}$, can be obtained by using a general update rule $\boldsymbol{x}^{k+1}=f_k(\boldsymbol{x}^k,\mathcal{G}_k), k=0,1,2,\ldots$, where $f_k(\cdot)$ can be a general time-varying function, depending on the problem setup, and captures the interaction laws among the agents. Therefore, the main goal here is to understand whether the generated sequence of states $\{\boldsymbol{x}^k\}_{k=0}^{\infty}$ will converge (stabilize) to any equilibrium. That has been the subject of much research effort, including work in distributed control and computation \cite{nedic2018network}. Unfortunately, despite enormous efforts in the area, the stability problem for such dynamics in its full generality is still far from being solved. However, partial solutions to the problem under certain simplifying assumptions are known. For instance, there has been a rich body of literature on the analysis of multiagent network systems, mainly from the static point of view, in which a set of agents iteratively interact over a \emph{fixed} network to achieve a certain goal, such as consensus or optimization of an objective function. The classical models of DeGroot \cite{degroot1974reaching} and Friedkin-Johnsen \cite{friedkin1997social} in the social sciences are two special types of such systems \cite{nedic2018network,AVP-RT:17}. Below is a sample result in this area \cite{olfati2004consensus,levin2017markov}:

\begin{theorem}\label{thm-consensus}
Given a fixed and undirected connected network $\mathcal{G}=([n],\mathcal{E})$, at any time $k=0,1,\ldots$, let every agent $i$ take the average of its own state and those of its neighbors, i.e., $x_i^{k+1}=\sum_{j\in N_i\cup\{i\}}a_{ij}x_j^k, \forall i\in[n]$, where $a_{ij}>0$ are positive constant weights such that $\sum_{j=1}^{n}a_{ij}=1, \forall i$. Then the generated averaging dynamics are Lyapunov stable and converge to an equilibrium point.
\end{theorem}

By comparing Theorem \ref{thm-consensus} with the aforementioned general dynamics, one can identify several simplifying assumptions that have been made in most of the existing results on multiagent network systems. 1) The underlying networks are fixed as $\mathcal{G}_k=\mathcal{G}, \forall k$. 2) The underlying networks are connected and undirected. 3) The underlying networks $\mathcal{G}_k$ do not depend on the agents' states $\boldsymbol{x}^k$. Therefore, a large body of literature has been developed to establish the stability of the above general dynamics under less restrictive assumptions. The results in the static case can often be generalized to time-varying networks through the assumption of a certain ``independency" between the network process and the state dynamics. For instance, one of the commonly used assumptions is that the network dynamics are governed by an exogenous process that is uncoupled from the state dynamics \cite{olshevsky2009convergence,nedic2009distributed,nedic2015distributed,bacsar2016convergence,etesami2016convergence,etesami2017potential,zhu2011convergence}. Here is an extension from \cite{blondel2005convergence}:

\begin{theorem}\label{thm:hendrix}
Consider a sequence of time-varying directed graphs $\mathcal{G}_k=([n],\mathcal{E}_k)$ with the weight of edge $(i,j)$ at time $k$ being $a_{ij}(k)$. Assume that the sequence of graphs is $B$-strongly connected, meaning that for any $k\ge 0$, the graph $\mathcal{G}=([n], \cup_{s=k}^{k+B}\mathcal{E}_k)$ is strongly connected. Moreover, given a positive constant $\alpha\in (0,1)$, assume $a_{ij}(k)\in [\alpha, 1]\cup\{0\}, \forall i,j,k$, and $a_{ii}(k)\ge \alpha, \ \sum_{j}a_{ij}(k)=1, \forall i,k$. Then the averaging dynamics $x_i^{k+1}=\sum_{j\in N_i}a_{ij}(k)x_j^k, \ i\in[n]$, will converge to a consensus point.
\end{theorem}

While Theorem \ref{thm:hendrix} relaxes the static communication network to time-varying networks, it still has shortcomings in addressing many realistic multiagent systems. For instance, the network connectivity must be preserved over any time window of length $B$, and it is hard to check whether that is happening (especially if the networks are generated endogenously based on the agents' states). Second, the assumptions on the weight matrices are somewhat restrictive as, in real situations, the weights can approach $0$ and then increase again to $1$. Besides, the theorem uses an implicit assumption on the symmetry of the networks by imposing strong connectivity. Finally, in realistic situations, the evolution of the network itself depends on the evolution of agents' states. In contrast, in the above theorem, the network dynamics are driven by an exogenous process that is independent of how the states evolve. A generalization of Theorem \ref{thm:hendrix} is to allow weak coupling between the state and network dynamics, with certain network connectivity/symmetry assumptions \cite{olfati2004consensus,hendrickx2013convergence,jadbabaie2003coordination}. We refer to \cite{sonin2008decomposition,touri2011existence} for other extensions of such results that use a backward product of stochastic matrices. We mention here that our work is also related to dynamic clustering, for which the goal is to provide a theoretical justification for cluster synchronization in multiagent systems by using saddle-point dynamics \cite{burger2012hierarchical,burger2014duality}. However, the network structure in that application is fixed and captured by a set of linear constraints. In contrast, in our work, the network dynamically evolves as a complex function of the state variables. 

While the existing results can address a large class of multiagent network systems, there are still many examples that do not fit into any of the categories mentioned above or for which the application of the above techniques provides poor results on the behavior of the agents. Our work is fundamentally different from the earlier literature and offers a new perspective for the averaging dynamics by capturing the internal co-evolution of state and network dynamics. Thus, in this paper, we depart from conventional methods for stability analysis of multiagent dynamics, such as Markov chains or products of stochastic matrices. Our approach allows us to relax some of the common assumptions, such as global knowledge on the network connectivity throughout the dynamics. We believe that this new approach, together with the current results on the multiagent averaging dynamics, can be used to analyze a broader class of complex state-dependent network dynamics.    

\subsection{Contributions and Organization} 

Inspired by the above shortcomings and building upon our previous work \cite{etesami2019simple}, in this paper, we provide a principled framework from an optimization perspective to study the Lyapunov stability of multiagent state-dependent network dynamics. We show that despite the challenges due to state-network coupling, it is still possible to capture the co-evolution of network and state dynamics for a broad class of multiagent systems, even under an asymmetric or nonlinear environment. More precisely, we show that often the network structure among the agents can be viewed as \emph{dual variables} of a constrained optimization problem, where the existence of an edge is related to the tightness of the corresponding constraint. As a result, we can view multiagent network dynamics as an iterative primal-dual algorithm for a static constrained optimization problem where the primal updates correspond to state updates of the dynamics, and the dual updates correspond to the network evolution. The KKT optimality conditions also guide the coupling between the network and state dynamics. This approach allows us to view the constrained Lagrangian of the underlying static problem as a Lyapunov or ``semi-Lyapunov" function for the multiagent dynamics. Therefore, we obtain a principled way to establish the stability of multiagent network dynamics in terms of the asymptotic convergence of an iterative optimization algorithm. That means that a variety of iterative optimization methods can be used to study the stability of multiagent state-dependent network dynamics.

In Section \ref{sec:model}, we first provide our problem formulation, which involves modeling of a large class of state-dependent network dynamics. In Section \ref{sec:BCD}, we apply a sequential optimization framework based on the block coordinate descent method to establish Lyapunov stability for a large class of state-dependent network dynamics. We consider that method under both symmetric and asymmetric network structures and use the change of variables to generate other types of state-dependent network dynamics. In Section \ref{sec:Saddle}, we use a saddle-point model to extend our results to a case in which there is a conflict between the network structure and the state evolution. In Section \ref{sec:continuous-saddle}, we consider continuous-time dynamics for which the edge emergence between agents is no longer a binary event, but rather a continuous weight process. In Section \ref{sec:simu}, we provide numerical results, and we conclude the paper by identifying some future directions of research in Section \ref{sec:conclusion}.   

\subsection{Notation} 
\noindent
For a positive integer $n\in\mathbb{Z}^+$ we set $[n]:=\{1,2,\ldots,n\}$. We use bold symbols for vectors. Given a vector $\boldsymbol{v}\in \mathbb{R}^n$ we use $\boldsymbol{v}^T$ to denote its transpose, $\|\boldsymbol{v}\|^2=\boldsymbol{v}^T\boldsymbol{v}$ to denote its Euclidean norm, and $\|\boldsymbol{v}\|_1=\sum_{i=1}^{n}|v_i|$ to denote its $l_1$-norm. We let $\mbox{diag}(\boldsymbol{v})$ be a diagonal matrix with vector $\boldsymbol{v}$ as its diagonal elements and zero elsewhere. Given a positive-definite matrix $Q$, we let $\|\boldsymbol{v}\|^2_Q=\boldsymbol{v}^TQ\boldsymbol{v}$. We use $\mathcal{C}$ to denote the class of real-valued differentiable functions with a finite global minimum. Similarly, we let $\mathcal{C}^2$ be the class of twice  differentiable functions with a finite global minimum. Given a strictly convex function $\Psi\in \mathcal{C}$, we use $D_{\Psi}(\boldsymbol{x},\boldsymbol{y})=\Psi(\boldsymbol{x})-\Psi(\boldsymbol{y})-\nabla \Psi(\boldsymbol{y})^T(\boldsymbol{x}-\boldsymbol{y})$ to denote the Bregman divergence with respect to $\Psi$. Finally, $f:\mathbb{R}^n\to \mathbb{R}$ is called $L$-\emph{Lipschitz continuous} if there exists $L>0$ such that $|f(\boldsymbol{x})-f(\boldsymbol{y})|\leq L\|\boldsymbol{x}-\boldsymbol{y}\|, \forall \boldsymbol{x},\boldsymbol{y}\in\mathbb{R}^n$.

\section{Problem Formulation}\label{sec:model}

Let us consider a multiagent network system consisting of $[n]$ agents. At any given time $k=0,1,2,\ldots$, we use  $x_{i}^k\in \mathbb{R}$ to denote the state of agent $i$, and $\boldsymbol{x}^k=(x_1^k,\ldots,x_n^k)^T$ to denote the state of the entire system at that time.\footnote{For simplicity of presentation, we assume that agents' states are scalar real numbers. However, most of the results can be naturally extended to the case in which agents' states are vectors in $\mathbb{R}^d$.} Moreover, we assume that each agent $i\in [n]$ has $n-1$ measurement functions $g_{ij}(x_i,x_j):\mathbb{R}^2\to\mathbb{R}$, one for every other agent $j\in [n]\setminus\{i\}$. For most of this paper, we assume that the measurement functions $g_{ij}$ are twice differentiable and belong to $\mathcal{C}^2$. Given any state $\boldsymbol{x}$, we assume that the set of neighbors of an agent $i\in[n]$ is determined by the logic constraints $g_{ij}(x_i,x_j)\leq 0, \ j\in [n]\setminus\{i\}$. In other words, for a given state $\boldsymbol{x}$, agent $i$ is influenced by agent $j$ (or $j$ is a neighbor of $i$) if and only if $g_{ij}(x_i,x_j)\leq 0$. In particular, we use $N_i(\boldsymbol{x}):=\{j: g_{ij}(x_i,x_j)\leq 0\}$ to denote the set of neighbors of agent $i$ at a state $\boldsymbol{x}$. At any time instance $k$, each agent $i\in [n]$ interacts with its neighbors and updates its state at the next time step to
\begin{align}\label{eq:state-dependent-model}
x_i^k=\phi_i\big(\boldsymbol{x}^k, N_i(\boldsymbol{x}^k)\big), \ \ i\in [n],
\end{align}
where $\phi_i(\cdot)$ is an agent-specific update rule, which is a function of the states of agent $i$'s neighbors. Note that the above discrete-time dynamics contain a broad class of state-dependent network dynamics for which the network at time $k$ is given by $\mathcal{G}_k=([n], \{(i,j): j\in N_i(\boldsymbol{x}^k)\})$. It is evident that the network structure at time $k$ depends on the agents' states at that time, and the state at the next time step $k+1$ is a function of the network structure at the current time $k$.

\begin{definition}
The measurement functions $g_{ij}(\cdot)$ are called \emph{symmetric} if for all $i\neq j$ we have $g_{ij}(x_i,x_j)=g_{ji}(x_j,x_i)$. Note that for symmetric measurement functions, the communication network $\mathcal{G}_k$ at any time $k$ is an undirected graph.  
\end{definition}

One of our main objectives in this paper is to provide a general class of update rules $\phi_i(\cdot)$ such that the state-dependent network dynamics \eqref{eq:state-dependent-model} converge to some equilibrium point or are Lyapunov stable in the following sense:

\begin{definition}
A function $V:\mathbb{R}^{n}\to \mathbb{R}$ is called a \emph{Lyapunov function} for the discrete time dynamical system $\boldsymbol{z}^{k+1}=h_k(\boldsymbol{z}^k), \ k=0,1,2,\ldots$, if it is decreasing along the trajectories of the dynamics, i.e., $V(\boldsymbol{z}^{k+1})< V(\boldsymbol{z}^k), \forall k$. We refer to a dynamical system that admits a Lyapunov function as \emph{Lyapunov stable}.
\end{definition}

\begin{remark}
While a Lyapunov function $V(\cdot)$ is typically defined to be a nonnegative function with $V(0)=0$, as we shall see, all the Lyapunov functions in this paper are bounded below by some global constant $M$. Therefore, if those functions are shifted by a constant $|M|$, it is ensured that $V(\cdot)+|M|$ is nonnegative and strictly decreasing along the trajectory of the dynamics. Moreover, we do not require $V(0)=0$, as the origin is not necessarily an equilibrium point of the dynamics.
\end{remark}

To illustrate the generality of the above model, let us consider the well-known homogeneous Hegselmann-Krause (HK) model from social science \cite{hegselmann2002opinion}. In the homogeneous HK model, there is a set of $[n]$ agents, and it is assumed that at each time instance $k=0, 1, 2, \ldots$, the opinion (state) of agent $i\in[n]$ can be represented by a scalar $x_{i}^k\in \mathbb{R}$. Each agent $i$ updates its state at time $k+1$ by taking the arithmetic average of its state and those of others that are in its $\epsilon$-neighborhood at time $k$, i.e., 
\begin{align}\nonumber
x_i^{k+1}=\frac{x_i^k+\sum_{j\in N_i(\boldsymbol{x}^k)} x_j^k}{1+|N_i(\boldsymbol{x}^k)|}, \ \ \ \ i\in [n].
\end{align}
Here $\epsilon>0$ is a constant parameter, and $N_i(\boldsymbol{x}^k)=\{j\in[n]\setminus\{i\}: |x_i^k-x_j^k|\leq \epsilon\}$ denotes the set of neighbors of agent $i$ at time $k$. In fact, it is known that such dynamics are Lyapunov stable and converge to an equilibrium point \cite{hegselmann2002opinion,proskurnikov2018tutorial}. It is easy to see that homogeneous HK dynamics are a very special case of the state-dependent network dynamics \eqref{eq:state-dependent-model} for which the symmetric measurements are $g_{ij}(x_i,x_j)=\frac{(x_i-x_j)^2}{2}-\frac{\epsilon^2}{2}$, and the update rule is given by $\phi_i\big(\boldsymbol{x}, N_i(\boldsymbol{x})\big)=\frac{x_i+\sum_{j\in N_i(\boldsymbol{x})} x_j}{1+|N_i(\boldsymbol{x})|}$.

\section{Lyapunov Stability Using Block Coordinate Descent}\label{sec:BCD}

A popular approach to solving optimization problems is the so-called \emph{block coordinate descent} (BCD) method, which is also known as the \emph{Gauss-Seidel method}.  At each iteration of this method, the objective function is minimized with respect to a single block of variables while the rest of the blocks are held \emph{fixed}. More specifically, consider this optimization problem: $\min \{F(\boldsymbol{y}_1,\ldots,\boldsymbol{y}_n), \ \boldsymbol{y}_i\in Y_i, \forall i\}$,
where $Y_i\subseteq \mathbb{R}^{m_i}$ is a closed convex set, and $F: \prod_{i=1}^{n}Y_i\to \mathbb{R}$ is a continuous function. At iteration $t=0,1,\ldots$ of the BCD method, the block variable $\boldsymbol{y}_i$ is updated through solving of the subproblem $\boldsymbol{y}_i^t=\arg\min_{\boldsymbol{z}_i\in Y_i} F(\boldsymbol{y}^{t}_1,\ldots,\boldsymbol{y}_{i-1}^{t},\boldsymbol{z}_i,\boldsymbol{y}_{i+1}^{t},\ldots,\boldsymbol{y}^t_n), \ i\in[n]$. Since, in practice, finding the exact minimum in each iteration might be difficult, one can consider an \emph{inexact} BCD method, whereby a smooth regularizer is added to the objective function or is approximated by a simpler convex function. In either case, and under some mild assumptions, it can be shown that the inexact BCD method will converge to a stationary point of the objective function $F(\cdot)$.

Now let us consider the following constrained nonlinear program:
\begin{align}\nonumber
&\min \ f(\boldsymbol{x}):=\sum_{i=1}^{n}f_i(x_i)\cr 
&\mbox{s.t.} \ \ \ \ \ g_{ij}(x_i,x_j)\leq 0, \ \forall i\neq j, \cr 
& \ \ \ \qquad \boldsymbol{x}\in \mathbb{R}^n,
\end{align} 
where $f_i\in \mathcal{C}^2, \forall i\in [n]$ and $g_{ij}\in \mathcal{C}^2 , \forall i\neq j$ are differentiable measurement functions between the two members of each pair of agents. In other words, each agent $i\in[n]$ has a private function $f_i(x_i)$, and the agents collectively want to choose their states to minimize the global objective function $f(\boldsymbol{x}):=\sum_{i=1}^{n}f_i(x_i)$ while they all remain connected. If we dualize the constraints by using dual variables $\lambda_{ij}\ge 0$, and form the Lagrangian function, we have,
\begin{align}\nonumber
L(\boldsymbol{x},\boldsymbol{\lambda})=f(\boldsymbol{x})+\sum_{i\neq j}\lambda_{ij}g_{ij}(x_i,x_j),
\end{align}
which is a function consisting of two block variables, namely a \emph{state} block variable $\boldsymbol{x}:=(x_1,\ldots,x_n)\in \mathbb{R}^{n}$, and a nonnegative \emph{network} block variable $\boldsymbol{\lambda}:=(\lambda_{ij}, i\neq j)$. Now let us minimize the Lagrangian function by using the BCD method subject to the box constraints $\lambda_{ij}\in [0,1], \forall i,j$. As $L(\boldsymbol{x},\boldsymbol{\lambda})$ is a linear function of $\boldsymbol{\lambda}$, if we fix the state block variable and minimize $L(\boldsymbol{x},\boldsymbol{\lambda})$ with respect to $\boldsymbol{\lambda}\in [0,1]^{n(n-1)}$, we get $\lambda_{ij}=1$ if $g_{ij}(x_i,x_j)\leq 0$ (i.e., there is a directed edge from agent  $i$ to agent $j$), and $\lambda_{ij}=0$ if $g_{ij}(x_i,x_j)> 0$ (i.e., no such an edge exists). Thus, fixing the state variable and minimizing the Lagrangian with respect to $\boldsymbol{\lambda}\in [0,1]^{n(n-1)}$, the dual variables precisely capture the network structure among the agents for that state. 

Motivated by many applications of distributed averaging over networks, such as for consensus \cite{olfati2004consensus}, opinion dynamics \cite{hegselmann2002opinion,degroot1974reaching}, distributed optimization \cite{nedic2009distributed,nedic2015distributed}, and formation control \cite{bullo2009distributed}, in this paper, we provide several classes of distributed averaging dynamics over complex state-dependent networks. As we shall see, these dynamics can not only recover several types of well-known linear averaging dynamics from the physical and social sciences, but also be extended to \emph{nonlinear} averaging dynamics over state-dependent network topologies. The following theorem provides our first class of nonlinear averaging dynamics whose specification to quadratic measurements can recover several linear averaging dynamics.

\begin{theorem}\label{thm:majorizing}
Let $f_i(x_i)\in \mathcal{C}^2$ and $g_{ij}(x_i,x_j)\in \mathcal{C}^2$ be symmetric functions with $|\frac{\partial^2 g_{ij}}{\partial x_i \partial x_j}|\leq m, |\frac{\partial^2f_{i}}{\partial x^2_i}|\leq m, \forall i,j$.\footnote{For instance, any $m$-smooth function possesses this property.} Then the state-dependent network dynamics
\begin{align}\label{eq:majorizing-dynamics}
x_i^{k+1}=x_i^k-\frac{\frac{\partial}{\partial x_i}f_i(x^k_i)+\sum_{j\in N_i(\boldsymbol{x}^k)} \frac{\partial}{\partial x_i}g_{ij}(x^k_i,x^k_j)}{2m(|N_i(\boldsymbol{x}^k)|+1)}, \ \  i\in [n]
\end{align}
admit a Lyapunov function $V(\boldsymbol{x}):=\sum_i f_i(x_i)+\frac{1}{2}\sum_{i,j}\min\{g_{ij}(x_i,x_j),0\}$ such that $V(\boldsymbol{x}^{k+1})\leq V(\boldsymbol{x}^k)-m\|\boldsymbol{x}^k-\boldsymbol{x}^{k+1}\|^2$. If, in addition, $g_{ij}(\cdot)$ are convex and $f_i(\cdot)$ are strictly convex functions, then the dynamics \eqref{eq:majorizing-dynamics} will converge to an equilibrium point.  
\end{theorem}
\begin{proof}
Let us consider the following Lagrangian function 
\begin{align}\nonumber
L(\boldsymbol{x},\boldsymbol{\lambda})=\sum_{i=1}^n f_i(\boldsymbol{x}_i)+\frac{1}{2}\sum_{i\neq j}\lambda_{ij}g_{ij}(x_i,x_j),
\end{align}
and consider the BCD method applied to this function when $\boldsymbol{\lambda}\in [0,1]^{n(n-1)}$ and $\boldsymbol{x}\in \mathbb{R}^n$. If the state variable is fixed to $\boldsymbol{x}^k$, and we set $\boldsymbol{\lambda}^k:=\argmin_{\boldsymbol{\lambda}\in [0,1]^{n(n-1)}}L(\boldsymbol{x}^k,\boldsymbol{\lambda})$, it is easy to see that $\boldsymbol{\lambda}^k$ precisely captures the network structure among the agents at the current state $\boldsymbol{x}^k$. Next, let us fix the network variable to $\boldsymbol{\lambda}^k$, and consider 
\begin{align}\nonumber
L_k(\boldsymbol{x}):=L(\boldsymbol{x},\boldsymbol{\lambda}^k)&=\sum_{i=1}^n f_i(\boldsymbol{x}_i)+\frac{1}{2}\sum_{i\neq j}\lambda^k_{ij}g_{ij}(x_i,x_j)\cr 
&=\sum_{i=1}^n f_i(\boldsymbol{x}_i)+\frac{1}{2}\sum_i\sum_{j\in N_i(\boldsymbol{x}^k)} g_{ij}(x_i,x_j).
\end{align}    
Ideally, we want to set the state at the next time step $\boldsymbol{x}^{k+1}$ to a global minimizer of $L_k(\boldsymbol{x})$. However, since it might be difficult to solve the minimization problem $\min_{\boldsymbol{x}\in\mathbb{R}^n}L_k(\boldsymbol{x})$, we use an inexact BCD method, with which a quadratic upper approximation of $L_k(\boldsymbol{x})$ is minimized. More precisely, consider the quadratic approximation of $L_k(\boldsymbol{x})$ at the current point $\boldsymbol{x}^k$:
\begin{align}\label{eq:upper-appx-quad}
L_k(\boldsymbol{x})\simeq L(\boldsymbol{x}^k)+(\boldsymbol{x}-\boldsymbol{x}^k)^T\nabla L_k(\boldsymbol{x}^k)+\frac{1}{2}(\boldsymbol{x}-\boldsymbol{x}^k)^T\nabla^2 L_k(\boldsymbol{x}^k)(\boldsymbol{x}-\boldsymbol{x}^k),
\end{align}  
where $\nabla L_k(\boldsymbol{x}^k)$ is the gradient of $L_k(\boldsymbol{x})$ at $\boldsymbol{x}^k$, whose $i$th component is given by 
\begin{align}\nonumber
[\nabla L_k(\boldsymbol{x}^k)]_i=\frac{\partial}{\partial x_i}f_i(x_i)+\sum_{j\in N_i(\boldsymbol{x}^k)} \frac{\partial}{\partial x_i}g_{ij}(x^k_i,x^k_j).
\end{align}
Moreover, $\nabla^2 L_k(\boldsymbol{x}^k)$ is the Hessian of $L_k(\boldsymbol{x})$ at $\boldsymbol{x}^k$, with the Hessian matrix function
\begin{align}\nonumber
[\nabla^2 L_k(\boldsymbol{x})]_{ij}=\begin{cases}
\frac{\partial^2}{\partial x^2_i}f_i(x_i)+\sum_{j\in N_i(\boldsymbol{x}^k)} \frac{\partial^2 g_{ij}(x_i,x_j)}{\partial x^2_i} & \mbox{if} \ j=i\\
\frac{\partial^2 g_{ij}(x_i,x_j)}{\partial x_i \partial x_j} & \mbox{if} \ j\in N_i(\boldsymbol{x}^k)\\
0 & \mbox{if} \ j\notin N_i(\boldsymbol{x}^k).
\end{cases}
\end{align}  
Now, if we assume $|\frac{\partial^2 g_{ij}}{\partial x_i \partial x_j}|\leq m, |\frac{\partial^2f_{i}}{\partial x^2_i}|\leq m, \forall i,j$, and use the Gershgorin Circle Theorem, one can see that for any $\boldsymbol{x}$, the Hessian $\nabla^2 L_k(\boldsymbol{x})$ is dominated by the diagonal matrix $Q_k:=2m\cdot \mbox{diag}(|N_1(\boldsymbol{x}^k)|+1,\ldots,|N_n(\boldsymbol{x}^k)|+1)$. Via the Tailor expansion, for every $\boldsymbol{x}\in \mathbb{R}^n$ there exists an $\boldsymbol{\zeta}_x\in\mathbb{R}^n$ such that,
\begin{align}\nonumber
L_k(\boldsymbol{x})&=L(\boldsymbol{x}^k)+(\boldsymbol{x}-\boldsymbol{x}^k)^T\nabla L_k(\boldsymbol{x}^k)+\frac{1}{2}(\boldsymbol{x}-\boldsymbol{x}^k)^T\nabla^2 L_k(\boldsymbol{\zeta}_x)(\boldsymbol{x}-\boldsymbol{x}^k)\cr 
&\leq L(\boldsymbol{x}^k)+(\boldsymbol{x}-\boldsymbol{x}^k)^T\nabla L_k(\boldsymbol{x}^k)+\frac{1}{2}(\boldsymbol{x}-\boldsymbol{x}^k)^TQ_k(\boldsymbol{x}-\boldsymbol{x}^k):=u_k(\boldsymbol{x}).
\end{align}  
Therefore, $u_k(\boldsymbol{x})$ is a quadratic upper approximation for $L_k(\boldsymbol{x})$ for any $\boldsymbol{x}$. Let
\begin{align}\label{eq:majorizing-compact-dynamics}
\boldsymbol{x}^{k+1}=\argmin_{\boldsymbol{x}\in \mathbb{R}^n} u_k(\boldsymbol{x})=\boldsymbol{x}^k-Q_k^{-1}\nabla L_k(\boldsymbol{x}^k).
\end{align}
We can write,
\begin{align}\nonumber
L_k(\boldsymbol{x}^{k+1})\leq u_k(\boldsymbol{x}^{k+1})\leq u_k(\boldsymbol{x}^k)=L_k(\boldsymbol{x}^k).
\end{align}
That shows that the state-dependent network dynamics,
\begin{align}\nonumber
\boldsymbol{x}_i^{k+1}=\boldsymbol{x}_i^k-\frac{\frac{\partial}{\partial x_i}f_i(x^k_i)+\sum_{j\in N_i(\boldsymbol{x}^k)} \frac{\partial}{\partial x_i}g_{ij}(x^k_i,x^k_j)}{2m(|N_i(\boldsymbol{x}^k)|+1)},
\end{align}
are Lyapunov stable (i.e., $L(\cdot)$ decreases regardless of the state or network updates), and the function 
\begin{align}\nonumber
V(\boldsymbol{x}):=\min_{\boldsymbol{\lambda}\in [0,1]^{n(n-1)}}L(\boldsymbol{x},\boldsymbol{\lambda})=\sum_i f_i(x_i)+\frac{1}{2}\sum_{i,j}\min\{g_{ij}(\boldsymbol{x}),0\}
\end{align}
serves as a Lyapunov function. Moreover, the drift of this Lyapunov is bounded by
\begin{align}\nonumber
V(\boldsymbol{x}^k)-V(\boldsymbol{x}^{k+1})&= L(\boldsymbol{x}^k,\boldsymbol{\lambda}^k)-L(\boldsymbol{x}^{k+1},\boldsymbol{\lambda}^{k+1})\ge L(\boldsymbol{x}^k,\boldsymbol{\lambda}^k)-L(\boldsymbol{x}^{k+1},\boldsymbol{\lambda}^{k})\cr 
&=L_k(\boldsymbol{x}^k)-L_k(\boldsymbol{x}^{k+1})=u_k(\boldsymbol{x}^k)-L_k(\boldsymbol{x}^{k+1})\cr 
&\ge u_k(\boldsymbol{x}^k)-u_k(\boldsymbol{x}^{k+1})= \frac{1}{2}(\nabla L_k(\boldsymbol{x}^k))^TQ_k^{-1}\nabla L_k(\boldsymbol{x}^k)\cr 
&=\frac{1}{2}\|Q_k^{-1}\nabla L_k(\boldsymbol{x}^k)\|^2_{Q_k}=\frac{1}{2}\|\boldsymbol{x}^k-\boldsymbol{x}^{k+1}\|^2_{Q_k}\ge m\|\boldsymbol{x}^k-\boldsymbol{x}^{k+1}\|^2,
\end{align}
where the last equality follows from \eqref{eq:majorizing-compact-dynamics}, and the last inequality holds as all the diagonal entries of $Q_k$ are greater than $2m$. Therefore, $V(\boldsymbol{x}^{k+1})\leq V(\boldsymbol{x}^k)-m\|\boldsymbol{x}^k-\boldsymbol{x}^{k+1}\|^2$. Since $V(\cdot)$ is lower bounded by a finite value, we get $\lim_{k\to \infty} \|\boldsymbol{x}^{k+1}-\boldsymbol{x}^k\|=0$. That, in view of \eqref{eq:majorizing-compact-dynamics} and the fact that diagonal entries of $Q_k^{-1}$ are lower bounded by $\frac{1}{2mn}$, implies $\lim_{k\to \infty}\nabla L_k(\boldsymbol{x}^k)=\boldsymbol{0}$.

Next, to show the convergence of the dynamics \eqref{eq:majorizing-dynamics} in the case of convex measurements $g_{ij}(\cdot)$ and strictly convex functions $f_i(\cdot)$, we note that for any $k$, $L_k(\boldsymbol{x})$ belongs to the following \emph{finite} family of strictly convex functions 
\begin{align}\nonumber
\mathcal{H}:=\Big\{\sum_if_i(x_i)+\frac{1}{2}\sum_{i,j}\lambda_{ij}g_{ij}(x_i,x_j): \lambda_{ij}\in\{0,1\}, \forall i,j\Big\},
\end{align}
which contains at most $O(2^{n^2})$ functions. The reason is that $L_k(\boldsymbol{x})=L(\boldsymbol{x},\boldsymbol{\lambda}^k)$, where $\boldsymbol{\lambda}^k$ is the solution of the linear program $\min_{\boldsymbol{\lambda}\in [0,1]^{n(n-1)}}L(\boldsymbol{x}^k,\boldsymbol{\lambda})$, and must be an extreme point of $[0,1]^{n(n-1)}$. Now, given any $h(\boldsymbol{x})\in \mathcal{H}$, let ${\rm h}_1<{\rm h}_2<\ldots$, be all the indicies $k$ for which $L_k(\boldsymbol{x})=h(\boldsymbol{x})$. Then we can partition the sequence $\{\boldsymbol{x}^k\}$ into at most $|\mathcal{H}|$ subsequences $\{\{\boldsymbol{x}^{{\rm h}_{\ell}}\}_{\ell\ge 1},h\in \mathcal{H}\}$. Since $\lim_{k\to \infty}\nabla L_k(\boldsymbol{x}^k)=\boldsymbol{0}$, that means that for any subsequence $\{\boldsymbol{x}^{{\rm h}_{\ell}}\}$ we have, $\lim_{\ell\to \infty}\nabla h(\boldsymbol{x}^{{\rm h}_{\ell}})=\lim_{\ell\to \infty}\nabla L_{{\rm h}_{\ell}}(\boldsymbol{x}^{{\rm h}_{\ell}})=\boldsymbol{0}$. As $h(\cdot)$ is a strictly convex function with a finite global minimum, the subsequence $\{\boldsymbol{x}^{{\rm h}_{\ell}}\}_{\ell\ge 1}$ must converge to the unique minimizer of $h(\cdot)$, denoted by $\boldsymbol{x}_h$. Since there are a finite number of such subsequences, for any $\epsilon>0$, there exists $K_{\epsilon}$ such that $\|\boldsymbol{x}^{{\rm h}_{\ell}}-\boldsymbol{x}_h\|<\epsilon, \forall h\in\mathcal{H}, \ell>K_{\epsilon}$. Let $\mathcal{X}=\{\boldsymbol{x}_h=\argmin h(\boldsymbol{x}): h\in \mathcal{H}\}$ be the finite set of minimizers of all the functions in $\mathcal{H}$, and choose $\epsilon:=\frac{1}{3}\min_{\boldsymbol{x}_p\neq \boldsymbol{x}_q\in \mathcal{X}}\|\boldsymbol{x}_p-\boldsymbol{x}_q\|$. Then for $\ell>K_{\epsilon}$, each subsequence $\{\boldsymbol{x}^{{\rm h}_{\ell}}\}_{\ell\ge 1}$ lies in an $\epsilon$-neighborhood of its limit point $\boldsymbol{x}_h$, and, moreover, there is no jump of the iterates between two distinct $\epsilon$-neighborhoods (Otherwise, $\|\boldsymbol{x}^{k+1}-\boldsymbol{x}^{k}\|>\frac{\epsilon}{3}$ for some $k$, contradicting the fact that $\lim_{k\to \infty} \|\boldsymbol{x}^{k+1}-\boldsymbol{x}^k\|=0$.) Thus, for $\ell>K_{\epsilon}$, all the subsequences $\{\{\boldsymbol{x}^{{\rm h}_{\ell}}\}_{\ell\ge 1}, h\in\mathcal{H}\}$ must lie in the same $\epsilon$-neighborhood, and hence the sequence $\{\boldsymbol{x}^k\}$ converges to a limit point $\boldsymbol{x}^*\in \mathcal{X}$.   
\end{proof}

\begin{example}\label{ex:exact}
Let $g_{ij}(x_i,x_j)=\frac{(x_i-x_j)^2}{2}-\frac{\epsilon^2}{2}, i\neq j$ be symmetric quadratic measurements and $f_i(x_i)=0, \forall i\in[n]$. Clearly, these functions satisfy the statement of Theorem \ref{thm:majorizing} with $m=1$. By applying Theorem \ref{thm:majorizing} directly to these functions, we obtain 
\begin{align}\nonumber
x_i^{k+1}=x_i^k-\frac{\sum_{j\in N_i(\boldsymbol{x}^k)} (x_i^k-x_j^k)}{2(|N_i(\boldsymbol{x}^k)|+1)}=\frac{|N_i|+2}{2(|N_i|+1)}x_i^k+\frac{1}{2}\frac{\sum_{j\in N_i(\boldsymbol{x}^k)}x_j^k}{|N_i(\boldsymbol{x}^k)|+1}, \ \ i\in [n],
\end{align}   
which has been shown to be Lyapunov stable. However, the above dynamics are not exactly the homogeneous HK dynamics, but rather a ``lazy" version of them wherein each agent puts a higher weight of (nearly) $\frac{1}{2}$ on its own state. The reason for losing a factor of $\frac{1}{2}$ for the \emph{general} (nonquadratic) measurements is that the quadratic upper approximation in \eqref{eq:upper-appx-quad} may not be exact, and the cost of such approximation is reflected by an extra factor of $\frac{1}{2}$ in the underlying \emph{nonlinear} dynamics \eqref{eq:majorizing-dynamics}. However, if the measurement functions are quadratic (as in the HK model), the quadratic approximation in \eqref{eq:upper-appx-quad} becomes exact, and one can skip the approximation step in \eqref{eq:majorizing-compact-dynamics} by directly computing the gradient and Hessian matrices in a closed form to show that $L_k(\boldsymbol{x}^{k+1})\leq L_k(\boldsymbol{x}^k)$. More precisely, for quadratic measurements we have
\begin{align}\nonumber
L_k(\boldsymbol{x})=L(\boldsymbol{x}^k)+(\boldsymbol{x}-\boldsymbol{x}^k)^T\nabla L_k(\boldsymbol{x}^k)+\frac{1}{2}(\boldsymbol{x}-\boldsymbol{x}^k)^T\nabla^2 L_k(\boldsymbol{x}^k)(\boldsymbol{x}-\boldsymbol{x}^k),
\end{align}
with a closed-form gradient $\nabla L_k(\boldsymbol{x}^k)=(D_k-A_k)\boldsymbol{x}^k$ and Hessian $\nabla^2 L_k(\boldsymbol{x}^k)=D_k+A_k$, where $A_k$ is the adjacency matrix of the communication network at state $\boldsymbol{x}^k$ and $D_k=\mbox{diag}(|N_1(\boldsymbol{x}^k)|,\ldots,|N_n(\boldsymbol{x}^k)|)$. If we take $Q_k=\mbox{diag}(|N_1(\boldsymbol{x}^k)|+1,\ldots,|N_n(\boldsymbol{x}^k)|+1)$ (i.e., without an extra factor of 2) and note that $\boldsymbol{x}^{k}-\boldsymbol{x}^{k+1}=Q_k^{-1}\nabla L_k(\boldsymbol{x}^k)$, we get
\begin{align}\nonumber
L_k(\boldsymbol{x}^{k+1})-L_k(\boldsymbol{x}^{k})&=(\boldsymbol{x}^{k+1}\!-\!\boldsymbol{x}^k)^T\nabla L_k(\boldsymbol{x}^k)+\frac{1}{2}(\boldsymbol{x}^{k+1}\!-\!\boldsymbol{x}^k)^T(D_k+A_k)(\boldsymbol{x}^{k+1}\!-\!\boldsymbol{x}^k)\cr 
&=\frac{1}{2}(\boldsymbol{x}^{k+1}\!-\!\boldsymbol{x}^k)^T\big[D_k+A_k-2Q_k\big](\boldsymbol{x}^{k+1}\!-\!\boldsymbol{x}^k)\leq -\|\boldsymbol{x}^{k+1}-\boldsymbol{x}^k\|^2,
\end{align}   
where the last inequality holds because $D_k+A_k-2Q_k=-2I-(D_k-A_k)\leq -2I$. Therefore, Lyapunov stability of the homogeneous HK dynamics 
\begin{align}\nonumber
\boldsymbol{x}^{k+1}=\boldsymbol{x}^{k}-Q_k^{-1}\nabla L_k(\boldsymbol{x}^k)=Q_k^{-1}(I+A_k)\boldsymbol{x}^k
\end{align}
can be viewed as a special case of Theorem \ref{thm:majorizing} for specific quadratic measurements with an associated Lyapunov function $V(\boldsymbol{x})=\sum_{i,j}\min\{\frac{(x_i-x_j)^2}{2}-\frac{\epsilon^2}{2},0\}$.
\end{example}

 \medskip
\begin{example}
Let $\mathcal{G}=([n],\mathcal{E})$ be a fixed undirected graph with a positive weight $a_{ij}=a_{ji}>0$ on each edge $\{i,j\}\in \mathcal{E}$. (We set $a_{ij}=a_{ji}=0$ if $\{i,j\}\notin \mathcal{E}$.) Assume that $\sum_{j\in N_i\cup\{i\}}a_{ij}=1, \forall i$, where $N_i$ denotes the fixed set of neighbors of agent $i$. Now let us define $f_i(x_i)=0, \forall i\in[n]$. Moreover, let $K>0$ be a very large constant. Consider the following symmetric measurements:
\begin{align}\nonumber
g_{ij}(x_i,x_j)=\begin{cases}\frac{a_{ij}}{2}(x_i-x_j)^2-K, \ &\mbox{if} \ \{i,j\}\in \mathcal{E}\\
 1 \ &\mbox{if} \ \{i,j\}\notin \mathcal{E}.
 \end{cases}
\end{align}
Clearly, these functions satisfy the statement of Theorem \ref{thm:majorizing}. As $K$ is chosen to be a very large number, regardless of the state $\boldsymbol{x}^k$ of the dynamics (to be defined soon), we always have $N_i(\boldsymbol{x}^k)=\{j: g_{ij}(x_i^k,x_j^k)\leq 0\}=N_i$. If we leverage the quadratic structure of the measurement functions, and use the exact approach as in Example \ref{ex:exact}, we obtain $\nabla L_k(\boldsymbol{x}^k)=(I-D-A)\boldsymbol{x}^k$ and $\nabla^2 L_k(\boldsymbol{x}^k)=I-D+A$, where $A=(a_{ij})$ is the weighted adjacency matrix of the graph $\mathcal{G}$ (with zero diagonal entries), and $D=\mbox{diag}(a_{11},\ldots,a_{nn})$ is the diagonal matrix of self-degrees. Therefore, $Q_k:=I$ is a diagonal matrix dominating the Hessian matrix (as $I-D+A\leq I$ by the Gershgorin Theorem and since $a_{ii}=1-\sum_{j\in N_i}a_{ij}, \forall i$). Thus, the dynamics   
\begin{align}\nonumber
\boldsymbol{x}^{k+1}=\boldsymbol{x}^{k}-Q_k^{-1}\nabla L_k(\boldsymbol{x}^k)=\boldsymbol{x}^{k}-(I-D-A)\boldsymbol{x}^k=(D+A)\boldsymbol{x}^k
\end{align}  
are Lyapunov stable with a Lyapunov function 
\begin{align}\nonumber
V(x)=\frac{1}{2}\sum_{i,j}\min\{g_{ij}(x_i,x_j), 0\}=-|\mathcal{E}|K+\frac{1}{2}\sum_{\{i,j\}\in \mathcal{E}}a_{ij}(x_i-x_j)^2 .
\end{align}
That is exactly the well-known Laplacian Lyapunov function for the conventional averaging dynamics $\boldsymbol{x}^{k+1}=(D+A)\boldsymbol{x}^k$, which is shifted by a constant $-|\mathcal{E}|K$. Hence, we recover the result of Theorem \ref{thm-consensus}. 
\end{example}

 \medskip
\begin{example}
In the proof of Theorem \ref{thm:majorizing}, we restricted our attention to quadratic upper approximations. However, motivated by the mirror descent algorithm from convex optimization \cite{bubeck2015convex}, we can use any smooth convex mirror map $\Psi:\mathbb{R}^n\to \mathbb{R}$ to construct an upper approximation for $L_k(\boldsymbol{x})$ at the point $x^k$. By doing so, we obtain alternative state-dependent network dynamics whose Lyapunov stability and convergence can be established via an approach similar to that in Theorem \ref{thm:majorizing}. More precisely, Let $\Psi\in \mathcal{C}$ be a strictly convex function with $\nabla^2 \Psi(\boldsymbol{x})\ge 2mn I$. Then, $u_k(\boldsymbol{x}):=L(\boldsymbol{x}^k)+(\boldsymbol{x}-\boldsymbol{x}^k)^T\nabla L_k(\boldsymbol{x}^k)+D_{\Psi}(\boldsymbol{x},\boldsymbol{x}^k)$ serves as a convex upper approximation for $L_k(\boldsymbol{x})$; hence, if we update the state at the next time step to $\boldsymbol{x}^{k+1}=\argmin_{\boldsymbol{x}\in\mathbb{R}^n} u_k(\boldsymbol{x})$, or equivalently to the solution of  
\begin{align}\label{eq:mirror-map}
\nabla \Psi(\boldsymbol{x}^{k+1})=\nabla \Psi(\boldsymbol{x}^{k})-\nabla L_k(\boldsymbol{x}^{k}),
\end{align} 
that will guarantee a decrease in the Lyapunov $V(\boldsymbol{x})\!=\!\sum_i f_i(x_i)+\frac{1}{2}\sum_{i,j}\min\{g_{ij}(\boldsymbol{x}),0\}$. For instance, if we choose the mirror map to be the negative entropy function, i.e., $\Psi(\boldsymbol{x}):=\sum_{i=1}^{n}x_i\ln x_i$, and use \eqref{eq:mirror-map}, we obtain the following Lyapunov stable multiplicative dynamics:
\begin{align}\nonumber
x_i^{k+1}=x_i^k\cdot\exp \big(-\frac{\partial f(x_i)}{\partial x_i}-\!\!\!\!\!\sum_{j\in N_i(\boldsymbol{x}^k)} \!\!\!\frac{\partial g_{ij}(x^k_i,x^k_j)}{\partial x_i}\big).
\end{align}    
\end{example}

\subsection{Asymmetric State-Dependent Network Dynamics}
Asymmetric (directed) interconnections among  agents often introduce a significant challenge in the analysis of multiagent network dynamics. Unfortunately, the gradient operator is ``symmetric," meaning that fixing the network variable in the BCD method and updating the state variable in the negative direction of the gradient will always generate a symmetric class of averaging dynamics. However, one way to tackle that issue using sequential optimization is to introduce an independent copy of the state variable while making sure that the two copies remain close to each other. In other words, we capture the asymmetry between the agents by introducing an extra block variable into the BCD method and adding a penalty term (possibly asymmetric) to the objective function. That ensures that the two copies of the state variables remain close to each other. Here, the choice of the penalty function can be very problem-specific, resulting in different asymmetric state-dependent network dynamics. However, one natural choice for the penalty function is the Bregman divergence between the two copies of the state variables, as shown in the following theorem.
  
\begin{theorem}\label{thm:asymmetric}
Let $f_i(x_i)\in \mathcal{C}^2$ and $g_{ij}(x_i,x_j)\in \mathcal{C}^2$ be $L$-Lipschitz continuous functions such that $|\frac{\partial^2 g_{ij}}{\partial x_i \partial x_j}|\leq m, |\frac{\partial^2f_{i}}{\partial x^2_i}|\leq m, \forall i,j$. If $\|\boldsymbol{y}-\boldsymbol{x}\|_1\leq \frac{1}{nL}D_{f}(\boldsymbol{y},\boldsymbol{x}), \forall \boldsymbol{x},\boldsymbol{y}$, where $f(\boldsymbol{x})=\sum_{i=1}^nf_i(x_i)$, then the dynamics
\begin{align}\label{eq:asymmetric}
x_i^{k+1}=x_i^{k}-\frac{\sum_{j\in N_i(\boldsymbol{x}^k)}\frac{\partial}{\partial x_i}g_{ij}(x_i^k,x_j^k)}{m(|N_i(\boldsymbol{x}^k)|+1)}, \ \ \ i\in[n],
\end{align}  
are Lyapunov stable with a Lyapunov function $V(\boldsymbol{x})=\sum_{i,j}\min\{g_{ij}(x_i,x_j),0\}$.
\end{theorem}
\begin{proof}
Let $c_i(\boldsymbol{\boldsymbol{x}},\boldsymbol{\lambda}_i):=\sum_{j\neq i}\lambda_{ij}g_{ij}(x_i,x_j)$ denote the cost of agent $i$ with respect to its neighbors, and let $\boldsymbol{y}$ be an independent copy of the state variable $\boldsymbol{x}$. Consider the following function with three independent block variables $\boldsymbol{\lambda}\in [0,1]^{n(n-1)}, \boldsymbol{x},\boldsymbol{y}\in \mathbb{R}^n$:
\begin{align}\nonumber
L(\boldsymbol{y},\boldsymbol{x},\boldsymbol{\lambda}):=\sum_{i=1}^{n}c_i(y_i,\boldsymbol{\boldsymbol{x}}_{-i},\boldsymbol{\lambda}_i)+D_{f}(\boldsymbol{y},\boldsymbol{x})=\sum_{i=1}^{n}\sum_{j\neq i}\lambda_{ij}g_{ij}(y_i,x_j)+D_{f}(\boldsymbol{y},\boldsymbol{x}),
\end{align}
where $D_{f}(\boldsymbol{y},\boldsymbol{x}):=f(\boldsymbol{y})-f(\boldsymbol{x})-(\boldsymbol{y}-\boldsymbol{x})^T\nabla f(\boldsymbol{x})$. Note that here we no longer require the symmetry assumption, so that in general, $g_{ij}(x_i,x_j)\neq g_{ji}(x_j,x_i)$. The reason for introducing the Bregman distance $D_{f}(\boldsymbol{y},\boldsymbol{x})$ into the objective function is that it would be ideal if the two copies of the state variables coincide. But instead of adding the hard constraint $\boldsymbol{y}=\boldsymbol{x}$ into our optimization problem, we relax this constraint by adding a soft penalty term to the objective function. Now, let us apply the BCD method to the following minimization:
\begin{align}\label{eq:three-block-objective}
\min_{\boldsymbol{\lambda}\in [0,1]^{n(n\!-\!1)}}\min_{\boldsymbol{x}\in\mathbb{R}^n}\min_{\boldsymbol{y}\in\mathbb{R}^n}\{\sum_{i=1}^{n}\sum_{j\neq i}\lambda_{ij}g_{ij}(y_i,x_j)+D_{f}(\boldsymbol{x},\boldsymbol{y})\}.
\end{align}  
First, assume that both state variables are fixed to $\boldsymbol{y}=\boldsymbol{x}=\boldsymbol{x}^k$. Then, if we minimize the objective function \eqref{eq:three-block-objective} with respect to $\boldsymbol{\lambda}\in [0,1]^{n(n\!-\!1)}$, we precisely capture the asymmetric network structure $\boldsymbol{\lambda}^k$ associated with the state $\boldsymbol{x}^k$. (I.e., $\lambda^k_{ij}=1$ if and only if $g_{ij}(x_i^k,x_j^k)\leq 0$.)     
Next, let us fix $\boldsymbol{\lambda}=\boldsymbol{\lambda}^k$ and $\boldsymbol{x}=\boldsymbol{x}^k$, and consider minimizing \eqref{eq:three-block-objective} with respect to the $\boldsymbol{y}$ variable. However, to obtain a closed form for the optimal solution, instead of solving the minimization exactly, we minimize its quadratic upper approximation at the current state $\boldsymbol{x}^k$, given by
\begin{align}\nonumber
z_k(\boldsymbol{y}):=L(\boldsymbol{x}^k,\boldsymbol{x}^k,\boldsymbol{\lambda}^k)+(\boldsymbol{y}-\boldsymbol{x}^k)^T\nabla_y L(\boldsymbol{x}^k,\boldsymbol{x}^k,\boldsymbol{\lambda}^k)+\frac{1}{2}(\boldsymbol{y}-\boldsymbol{x}^k)^TP_k(\boldsymbol{y}-\boldsymbol{x}^k),
\end{align}
where $P_k$ is a diagonal matrix whose $i$th diagonal entry is $m(|N_i(\boldsymbol{x}^k)|+1)$. To see why $L(\boldsymbol{y},\boldsymbol{x}^k,\boldsymbol{\lambda}^k)\leq z_k(\boldsymbol{y}), \forall \boldsymbol{y}$, we note that 
\begin{align}\nonumber
[\nabla_y L(\boldsymbol{y},\boldsymbol{x}^k,\boldsymbol{\lambda}^k)]_i=\sum_{j\in N_i(\boldsymbol{x}^k)}\frac{\partial}{\partial y_i}g_{ij}(y_i,x^k_j)+(\frac{\partial}{\partial y_i}f_i(y_i)-\frac{\partial}{\partial x_i}f_i(x_i^k)),
\end{align}
which implies that $\nabla^2 L(\boldsymbol{y},\boldsymbol{x}^k,\boldsymbol{\lambda}^k)$ is a \emph{diagonal} matrix whose diagonal entries are   
\begin{align}\nonumber
[\nabla^2 L(\boldsymbol{y},\boldsymbol{x}^k,\boldsymbol{\lambda}^k)]_{ii}=\frac{\partial^2}{\partial y^2_i}f_i(y_i)+\sum_{j\in N_i(\boldsymbol{x}^k)}\frac{\partial g_{ij}(y_i,x^k_{j})}{\partial^2 y_i}.
\end{align}
As $|\frac{\partial^2 f_i(y_i)}{\partial y^2_i}|\leq m$ and $|\frac{\partial g_{ij}(y_i,x^k_{j})}{\partial y^2_i}|\leq m$, for all $i\in [n]$, the Hessian matrix is dominated by $P_k$, and the result follows from the Tailor expansion. Thus, the optimal solution to $\min_{\boldsymbol{y}\in \mathbb{R}^n} z_k(\boldsymbol{y})$ is given by $\boldsymbol{x}^k-P_k^{-1}\nabla_y L(\boldsymbol{x}^k,\boldsymbol{x}^k,\boldsymbol{\lambda}^k)$, which is precisely the next state of the dynamics \eqref{eq:asymmetric}. Therefore, if the block variable $\boldsymbol{y}$ is updated to $\boldsymbol{x}^{k+1}$, while the other variables are fixed to $\boldsymbol{\lambda}=\boldsymbol{\lambda}^k, \boldsymbol{x}=\boldsymbol{x}^k$, that will decrease the objective function because
\begin{align}\nonumber
L(\boldsymbol{x}^{k+1},\boldsymbol{x}^k,\boldsymbol{\lambda}^k)\leq z_k(\boldsymbol{x}^{k+1})=\min_{\boldsymbol{y}\in \mathbb{R}^n}z_k(\boldsymbol{y})\leq z_k(\boldsymbol{x}^{k})=L(\boldsymbol{x}^{k},\boldsymbol{x}^k,\boldsymbol{\lambda}^k).
\end{align} 
Finally, let us fix $\boldsymbol{y}=\boldsymbol{x}^{k+1}, \boldsymbol{\lambda}=\boldsymbol{\lambda}^k$ (which are the solutions to their corresponding sub-optimizations), and consider $\min_{\boldsymbol{x}\in \mathbb{R}^n}L(\boldsymbol{x}^{k+1},\boldsymbol{x},\boldsymbol{\lambda}^k)$. In particular, if we show that $L(\boldsymbol{x}^{k+1},\boldsymbol{x}^{k+1},\boldsymbol{\lambda}^k)\leq L(\boldsymbol{x}^{k+1},\boldsymbol{x}^k,\boldsymbol{\lambda}^k)$, that will complete the BCD loop and imply that $L(\boldsymbol{y},\boldsymbol{x},\boldsymbol{\lambda})$ is decreasing along the trajectory of the asymmetric dynamics \eqref{eq:asymmetric}. Here, the role of the penalty term in the objective function comes into play. More precisely, using the Lipschitz continuity and the fact that $D_f(\boldsymbol{x}^{k+1},\boldsymbol{x}^{k+1})=0$, we can write, 
\begin{align}\nonumber
L(\boldsymbol{x}^{k+1},\boldsymbol{x}^{k+1},\boldsymbol{\lambda}^k)&-L(\boldsymbol{x}^{k+1},\boldsymbol{x}^k,\boldsymbol{\lambda}^k)=-D_f(\boldsymbol{x}^{k+1},\boldsymbol{x}^k)\cr 
&\qquad+\sum_{i=1}^{n}\sum_{j\in N_i(\boldsymbol{x}^k)}\Big(g_{ij}(x^{k+1}_i,x^{k+1}_j)-g_{ij}(x^{k+1}_i,x^{k}_j)\Big)\cr 
&\leq -D_f(\boldsymbol{x}^{k+1},\boldsymbol{x}^k)+\sum_{i,j}|g_{ij}(x^{k+1}_i,x^{k+1}_j)-g_{ij}(x^{k+1}_i,x^{k}_j)|\cr 
&\leq -D_f(\boldsymbol{x}^{k+1},\boldsymbol{x}^k)+nL\sum_{j}|x^{k+1}_j-x^{k}_j|\cr 
&= -D_f(\boldsymbol{x}^{k+1},\boldsymbol{x}^k)+nL\|\boldsymbol{x}^{k+1}-\boldsymbol{x}^k\|_1<0,
\end{align} 
where the last inequality follows from the assumption on the choice of the Bregman map $f(\cdot)$. Therefore, $V(\boldsymbol{x}):=\min_{\boldsymbol{\lambda}\in [0,1]^{n(n-1)}}L(\boldsymbol{x},\boldsymbol{x},\boldsymbol{\lambda})=\sum_{i,j}\min\{g_{ij}(x_i,x_j),0\}$ is a decreasing function along the trajectories of the asymmetric dynamics \eqref{eq:asymmetric}. 
\end{proof}

\subsection{BCD Method with Change of Variables}\label{subsec:change-variable}
In this section, we show how a suitable change of block variables in the BCD method can generate new state-dependent network dynamics, whose Lyapunov stability can be established using the same approach as before. The change of variable can be applied to either the state or the network variable. However, in this section, we focus on a more compelling case, in which the change of variable is applied on the network variable; we only illustrate the idea of a change of variable for the state through the following simple example.   
\begin{example}
Let us recall the homogeneous HK model in which the state of agent $i$ at the next time step is updated to $x_i^{k+1}=\frac{x_i^k+\sum_{j\in N_i(\boldsymbol{x}^k)} x_j^k}{1+|N_i(\boldsymbol{x}^k)|}, i\in [n]$, where the neighborhood set of agent $i$ is given by $N_i(\boldsymbol{x}^k)=\{j\in[n]\setminus\{i\}: |x_i^k-x_j^k|\leq \epsilon\}$. As the dynamics are invariant with respect to translation of the states, without loss of generality, we may assume that $x_i^k> 1, \forall i, k$. Now let us define a new state variable by setting $y_i:=\ln (x_i), i\in [n]$. If we apply the HK dynamics on those new logarithmic states, we obtain $y_i^{k+1}=\frac{y_i^k+\sum_{j\in N_i(\boldsymbol{y}^k)} y_j^k}{1+|N_i(\boldsymbol{y}^k)|}$, with $N_i(\boldsymbol{y}^k)=\{j\in[n]\setminus\{i\}: |y_i^k-y_j^k|\leq \epsilon\}$, which are Lyapunov stable and converge to an equilibrium. If we rewrite the dynamics in terms of $x$ variables, we obtain a new class of state-dependent geometric averaging dynamics $x_i^{k+1}=(\prod_{j\in \bar{N}_i(\boldsymbol{x}^k)}x_j^k)^{\frac{1}{|\bar{N}_i(\boldsymbol{x}^k)|}}$, with a new definition of a neighborhood set $\bar{N}_i(\boldsymbol{x}^k)=\{j\in[n]:  e^{-\epsilon}\leq \frac{x_i^k}{x_j^k}\leq e^{\epsilon}\}$, that are also Lyapunov stable and converge.
\end{example}

Next, we turn our attention to the case of changing the network variables. So far, the dual variable $\lambda_{ij}$ in the Lagrangian function $L(\boldsymbol{x},\boldsymbol{\lambda})$ has been used to capture the existence of an edge from agent $i$ to agent $j$. In particular, we saw that restricting $\lambda_{ij}$ to the unit interval $[0,1]$ and minimizing $L(\boldsymbol{x},\boldsymbol{\lambda})$ for the network variable in the BCD method would automatically force $\lambda_{ij}$ to take binary values in $\{1,0\}$ (and hence capture the switching behavior of an edge formation between $i$ and $j$). In particular, if we fix the state block variable in $L(\boldsymbol{x},\boldsymbol{\lambda})=\sum_{i,j}\lambda_{ij}g_{ij}(x_i,x_j)$ and minimize it with respect to $\boldsymbol{\lambda}\in [0,1]^{n(n-1)}$, that will give us the network structure for that state, i.e., $\lambda^*_{ij}=1-\mbox{sgn}(g_{ij}(x_i,x_j))$, where $\mbox{sgn}(\cdot)$ is the sign function. This relation will give us a complicated characterization of $\lambda^*_{ij}$ that is a combination of the sign function and the measurement function. An alternative way to recover the same network structure would start by transforming the original network variables from $\lambda_{ij}$ to $f_{ij}(\lambda_{ij}):=1-\mbox{sgn}(\lambda_{ij})$, which would cause the transformed Lagrangian function to become $\hat{L}(\boldsymbol{x},\boldsymbol{\lambda})=\sum_{i,j}f_{ij}(\lambda_{ij})g_{ij}(x_i,x_j)$. If we then apply the BCD method to $\hat{L}(\boldsymbol{x},\boldsymbol{\lambda})$ by fixing the state variable and optimizing $\hat{L}(\boldsymbol{x},\boldsymbol{\lambda})$ with respect to the \emph{unconstrained} network variable $\lambda_{ij}\in \mathbb{R}$, we obtain a simpler optimal network variable $\hat{\lambda}_{ij}=g_{ij}(x_i,x_j)$. Note that 
\begin{align}\nonumber
\hat{L}(\boldsymbol{x},\hat{\boldsymbol{\lambda}})=\min_{\boldsymbol{\lambda}\in \mathbb{R}^{n(n-1)}}\hat{L}(\boldsymbol{x},\boldsymbol{\lambda})=\min_{\boldsymbol{\lambda}\in [0,1]^{n(n-1)}}L(\boldsymbol{x},\boldsymbol{\lambda})=L(\boldsymbol{x},\boldsymbol{\lambda}^*).  
\end{align}
Such a transformation of the network variables has three advantages: i) it removes the box constraints on the network variables which cause switching and absorbs them into the structure of the transformation function; ii) the optimal network variable in the BCD method has a simpler form after the change of variable; and iii) by choosing different transfer functions, one can obtain different types of state-dependent network dynamics. The following theorem provides a sample result using the idea of a change of network variables.  

\begin{theorem}\label{thm:change-variable}
Let $g_{ij}(x_i,x_j)\in \mathcal{C}^2$ be symmetric functions such that $|\frac{\partial^2 g_{ij}}{\partial x_i\partial x_j}|\leq m, \forall i,j$. Moreover, let $f_{ij}(\lambda)\in \mathcal{C}$ be symmetric and nonnegative decreasing functions. Then the state-dependent network dynamics
\begin{align}\label{eq:change-variable-dynamics}
x_i^{k+1}=x_i^k-\frac{\sum_{j}f_{ij}\big(g_{ij}(x_i^k,x_j^k)\big)\frac{\partial}{\partial x_i}g_{ij}(x^k_i,x^k_j)}{2m\sum_{j}f_{ij}\big(g_{ij}(x_i^k,x_j^k)\big)}, \ \  i\in [n]
\end{align}
admit a Lyapunov function $V(\boldsymbol{x})\!=\!\sum_{i\neq j}\int_{0}^{g_{ij}(x_i,x_j)}f_{ij}(\lambda)d\lambda$. Moreover, for each $k$ there exists a positive-definite matrix $Q_k$ such that $V(\boldsymbol{x}^{k})-V(\boldsymbol{x}^{k+1})\ge \|\boldsymbol{x}^k-\boldsymbol{x}^{k+1}\|_{Q_k}^2$, and $\lim_{k\to \infty}\|\boldsymbol{x}^k-\boldsymbol{x}^{k+1}\|_{Q_k}=0$.
\end{theorem}   
\begin{proof}
Let us consider the BCD method applied to the transformed Lagrangian
\begin{align}\nonumber
\hat{L}(\boldsymbol{x},\boldsymbol{\lambda}):=\sum_{i\neq j}\Big(f_{ij}(\lambda_{ij})g_{ij}(x_i,x_j)-h_{ij}(\lambda_{ij})\Big),
\end{align}
with $\boldsymbol{x}\in \mathbb{R}^n$ and $\boldsymbol{\lambda}\in \mathbb{R}^{n(n-1)}$, for some real-valued functions $h_{ij}(\lambda_{ij})$ to be determined later. If for any fixed state $\boldsymbol{x}$, the function $\hat{L}(\boldsymbol{x},\boldsymbol{\lambda})$ has a unique minimum with respect to $\boldsymbol{\lambda}\in \mathbb{R}^{n(n-1)}$, we can apply the BCD method and be assured that the function will decrease because of the network updates. Now let us first fix the state variable to $\boldsymbol{x}^k$. Assuming that the functions $f_{ij}, h_{ij}$ are differentiable, to find $\argmin_{\boldsymbol{\lambda}\in \mathbb{R}^{n(n-1)}}\hat{L}(\boldsymbol{x}^k,\boldsymbol{\lambda})$, we set
\begin{align}\label{eq:phi-lamda-stationary}
\frac{\partial}{\partial \lambda_{ij}}\hat{L}(\boldsymbol{x}^k,\boldsymbol{\lambda})=f'_{ij}(\lambda_{ij})g_{ij}(x^k_i,x^k_j)-h'_{ij}(\lambda_{ij})=0,
\end{align}
which implies $\frac{h'_{ij}(\lambda_{ij})}{f'_{ij}(\lambda_{ij})}=g_{ij}(x^k_i,x^k_j)$. Therefore, if we define $h'_{ij}(\lambda):=\lambda f'_{ij}(\lambda)$, or, equivalently, $h_{ij}(\lambda):=\int_{0}^{\lambda}sf'_{ij}(s)ds$,  equation \eqref{eq:phi-lamda-stationary} has a unique solution $\lambda^*_{ij}=g_{ij}(x^k_i,x^k_j)$. To show that this solution is the minimizer of $\hat{L}(\boldsymbol{x}^k,\boldsymbol{\lambda})$, we note that
\begin{align}\nonumber
\frac{\partial}{\partial \lambda_{ij}}\hat{L}(\boldsymbol{x}^k,\boldsymbol{\lambda})=f'_{ij}(\lambda_{ij})\Big(g_{ij}(x^k_i,x^k_j)-\lambda_{ij}\Big).
\end{align}
Since $f_{ij}(\cdot)$ is a decreasing function, $f'_{ij}(\lambda_{ij})< 0$. Thus, for $\lambda_{ij}\leq g_{ij}(x^k_i,x^k_j)$, the function $\hat{L}(\boldsymbol{x}^k,\boldsymbol{\lambda})$ is decreasing with respect to $\lambda_{ij}$, and for $\lambda_{ij}\ge  g_{ij}(x^k_i,x^k_j)$, it is increasing. (Note that $\hat{L}(\boldsymbol{x}^k,\boldsymbol{\lambda})$ is splittable over its $\boldsymbol{\lambda}$-components, so we can analyze each of its summands separately.) Thus, given a fixed state $\boldsymbol{x}^k$, the unique global minimum of $\hat{L}(\boldsymbol{x}^k,\boldsymbol{\lambda})$ is obtained at $\boldsymbol{\lambda}=\boldsymbol{\lambda}^*$, where $\lambda^*_{ij}=g_{ij}(x^k_i,x^k_j)$. 

The rest of the proof follows along the same analysis used in Theorem \ref{thm:majorizing}, and we only sketch it here. Let us fix the network variable to $\boldsymbol{\lambda}^*$, and consider $\min_{\boldsymbol{x}\in \mathbb{R}^n} \hat{L}(\boldsymbol{x},\boldsymbol{\lambda}^*)$. To find a minimizer we use an inexact method by using a quadratic upper approximation of $\hat{L}(\boldsymbol{x},\boldsymbol{\lambda}^*)$ at $\boldsymbol{x}^k$. The $i$th component of the gradient of $\hat{L}(\boldsymbol{x},\boldsymbol{\lambda}^*)$ at $\boldsymbol{x}^k$ is given by 
\begin{align}\nonumber
[\nabla_{\boldsymbol{x}} \hat{L}(\boldsymbol{x}^k,\boldsymbol{\lambda}^*)]_i&=\sum_{j}\Big(f_{ij}(\lambda^*_{ij})\frac{\partial}{\partial x_i}g_{ij}(x^k_i,x^k_j)+f_{ji}(\lambda^*_{ji})\frac{\partial}{\partial x_i}g_{ji}(x^k_j,x^k_i)\Big)\cr 
&=2\sum_{j}f_{ij}(\lambda^*_{ij})\frac{\partial}{\partial x_i}g_{ij}(x^k_i,x^k_j)\cr 
&=2\sum_{j}f_{ij}(g_{ij}(x_i^k,x_j^k))\frac{\partial}{\partial x_i}g_{ij}(x^k_i,x^k_j),
\end{align}
where the second equality is by the symmetry of $f_{ij}, g_{ij}$. The Hessian matrix equals
\begin{align}\nonumber
[\nabla^2 \hat{L}(\boldsymbol{x},\boldsymbol{\lambda}^*)]_{ij}=\begin{cases}
2\sum_{j} f_{ij}(g_{ij}(x_i^k,x_j^k))\frac{\partial^2g_{ij}(x_i,x_j)}{\partial x^2_i} & \mbox{if} \ j=i\\
2f_{ij}(g_{ij}(x_i^k,x_j^k))\frac{\partial^2g_{ij}(x_i,x_j)}{\partial x_i \partial x_j} & \mbox{if} \ j\neq i,
\end{cases}
\end{align} 
which is dominated by a diagonal matrix $Q_k$ whose $i$th diagonal entry is given by $[Q_k]_{ii}=2+4m\sum_{j} f_{ij}(g_{ij}(x_i^k,x_j^k))$. Thus, the optimal solution to the quadratic upper approximation of $\hat{L}(\boldsymbol{x},\boldsymbol{\lambda}^*)$ at the point $\boldsymbol{x}=\boldsymbol{x}^k$ is given by $\boldsymbol{x}^k-Q_k^{-1}\nabla_{\boldsymbol{x}} \hat{L}(\boldsymbol{x}^k,\boldsymbol{\lambda}^*)$, which in view of \eqref{eq:change-variable-dynamics} equals to $\boldsymbol{x}^{k+1}$. As a result, $\hat{L}(\boldsymbol{x}^{k+1},\boldsymbol{\lambda}^*)< \hat{L}(\boldsymbol{x}^k,\boldsymbol{\lambda}^*)$, and
\begin{align}\nonumber
V(\boldsymbol{x})=\min_{\boldsymbol{\lambda}\in \mathbb{R}^{n(n-1)}}\hat{L}(\boldsymbol{x},\boldsymbol{\lambda})&=\sum_{i\neq j}\Big(f_{ij}(g_{ij}(x_i,x_j))g_{ij}(x_i,x_j)-\int_{0}^{g_{ij}(x_i,x_j)}\!\!\!\!\!\!\!\lambda f'_{ij}(\lambda)d\lambda\Big)\cr 
&=\sum_{i\neq j}\int_{0}^{g_{ij}(x_i,x_j)}\!\!\!\!\!\!f_{ij}(\lambda)d\lambda
\end{align}  
serves as a Lyapunov function for the dynamics \eqref{eq:change-variable-dynamics}, and the last equality is obtained using integration by parts. Now, as in Theorem \ref{thm:majorizing}, we can bound the drift of this Lyapunov function at time $k$ as
\begin{align}\nonumber
V(\boldsymbol{x}^{k})-V(\boldsymbol{x}^{k+1})&\ge \hat{L}(\boldsymbol{x}^{k},\boldsymbol{\lambda}^*)-\hat{L}(\boldsymbol{x}^{k+1},\boldsymbol{\lambda}^*)\cr 
&\ge \frac{1}{2}(\nabla_{\boldsymbol{x}} \hat{L}(\boldsymbol{x}^k,\boldsymbol{\lambda}^*))^TQ_k^{-1}\nabla_{\boldsymbol{x}} \hat{L}(\boldsymbol{x}^k,\boldsymbol{\lambda}^*)\cr 
&=\frac{1}{2}\|Q_k^{-1}\nabla_{\boldsymbol{x}} \hat{L}(\boldsymbol{x}^k,\boldsymbol{\lambda}^*)\|^2_{Q_k}\cr 
&=\frac{1}{2}\|\boldsymbol{x}^k-\boldsymbol{x}^{k+1}\|^2_{Q_k}.
\end{align}
Finally, $V(\cdot)$ is a nonnegative function due to the nonnegativity of $f_{ij}(\cdot)$. Therefore, if we sum the above relations for $k=1,2,\ldots$, we get $\sum_{k=1}^{\infty}\|\boldsymbol{x}^k-\boldsymbol{x}^{k+1}\|^2_{Q_k}\leq V(\boldsymbol{x}^0)<\infty$. This implies that $\lim_{k\to \infty}\|\boldsymbol{x}^k-\boldsymbol{x}^{k+1}\|_{Q_k}=0$.  
\end{proof}

\begin{remark}
In fact, the differentiability of the transfer functions $f_{ij}$ in the above theorem can be further relaxed to any nonnegative decreasing symmetric functions $f_{ij}$. In particular, Theorem \ref{thm:change-variable} is a heterogeneous extension of \cite[Corollary 1]{roozbehani2008lyapunov} (see also \cite{jabin2014clustering}) when $f_{ij}(\lambda)=f(\sqrt{\lambda}), \forall i,j$ and $g_{ij}(x_i,x_j)=(x_i-x_j)^2$. 
\end{remark}

\section{Stability Using Discrete-Time Saddle-Point Dynamics}\label{sec:Saddle}

In the previous section, we considered the stability of state-dependent network dynamics when the network structure and agents' states are aligned with each other. More precisely, in the application of the BCD method on the Lagrangian function $L(\boldsymbol{x},\boldsymbol{\lambda})$, we considered a double minimization problem $\min_{\boldsymbol{x}}\min_{\boldsymbol{\lambda}}L(\boldsymbol{x},\boldsymbol{\lambda})$, which essentially means that the network coordinator (viewed as a \emph{network player}) breaks/adds the links in favor of the agents' states (viewed as a \emph{state player}). In essence, that means that there is no conflict between the network and state players, as they are both minimizing the same Lagrangian function. But what if the network and state players have conflicting objectives? In that case, we have a 2-player zero-sum game between the network and the state with the payoff function $L(\boldsymbol{x},\boldsymbol{\lambda})$, so that the network player aims to maximize the payoff function while the state player aims to minimize it, i.e., $\min_{\boldsymbol{x}}\max_{\boldsymbol{\lambda}}L(\boldsymbol{x},\boldsymbol{\lambda})$. 
  
To model such conflicting behavior, we assume that each agent $i\in[n]$ holds $n-1$ \emph{convex} measurement functions $g_{ij}(x_i,x_j), j\in[n]\setminus\{i\}$. In this section, we restrict our attention to symmetric measurement functions; however, for asymmetric measurement functions, results similar to those in Theorem \ref{thm:asymmetric} can be obtained. For a given state $\boldsymbol{x}$, two agents $i$ and $j$ become each other's neighbors if $g_{ij}(x_i,x_j)\ge 0$. (Note that the logic constraint is now the reverse of what it was in the previous section.) Intuitively, an edge is formed between two agents $i$ and $j$ if and only if their states are far from each other. Now let us consider the following convex program:
\begin{align}\label{eq:convex-saddle-point}
&\min \ f(\boldsymbol{x}):=\sum_{i=1}^{n}f_i(x_i)\cr 
&\mbox{s.t.} \ \ \frac{1}{2}g_{ij}(x_i,x_j)\leq 0, \ \forall i\neq j, \cr 
&\qquad \ \boldsymbol{x}\in \mathbb{R}^n,
\end{align} 
where $f_i(x_i), i\in [n]$ are agents' private convex functions.\footnote{Here, each constraint is scaled by $\frac{1}{2}$ without causing changes to the feasible set.} To solve this problem, one can form the Lagrangian function $L(\boldsymbol{x},\boldsymbol{\lambda})=f(\boldsymbol{x})+\frac{1}{2}\sum_{i,j}\lambda_{ij}g_{ij}(x_i,x_j)$, and solve the following saddle-point problem: $\min_{\boldsymbol{x}\in \mathbb{R}^n}\max_{\boldsymbol{\lambda}\ge \boldsymbol{0}}L(\boldsymbol{x},\boldsymbol{\lambda})$.\footnote{To do so, one must find a solution $(\bar{\boldsymbol{x}},\bar{\boldsymbol{\lambda}})$ such that 
$L(\bar{\boldsymbol{x}},\boldsymbol{\lambda})\leq L(\bar{\boldsymbol{x}},\bar{\boldsymbol{\lambda}})\leq L(\boldsymbol{x},\bar{\boldsymbol{\lambda}}), \ \forall \boldsymbol{x}, \forall \boldsymbol{\lambda}\ge \boldsymbol{0}.$}

By using KKT optimality conditions, we know that if the constraint $g_{ij}(x_i,x_j)\leq 0$ is satisfied but not \emph{tight} (i.e., $g_{ij}(x_i,x_j)< 0$), then the corresponding optimal dual variable must be zero, i.e., $\lambda_{ij}=0$. If the dual variables are viewed as network variables, that means that there is no edge between the agents $i$ and $j$, and that is consistent with the logical condition of not having an edge between $i$ and $j$. On the other hand, if the constraint $g_{ij}(x_i,x_j)\leq 0$ is not satisfied (i.e., $g_{ij}(x_i,x_j)> 0$), then one must set the corresponding dual variable to $\lambda_{ij}=\infty$ to maximize $\max_{\boldsymbol{\lambda}\ge \boldsymbol{0}}L(\boldsymbol{x},\boldsymbol{\lambda})$. However, if the dual variables are upper-bounded by 1, to achieve the maximum value in $\max_{\boldsymbol{\lambda}\in [0,1]^{n(n-1)}}\mathcal{L}(\boldsymbol{x},\boldsymbol{\lambda})$, we must set $\lambda_{ij}$ to its upper bound, i.e., $\lambda_{ij}=1$. That is again consistent with the logical condition of having an edge between $i$ and $j$. These facts together suggest that the network switches that may occur during the update process of state-dependent network dynamics are merely the KKT optimality conditions that guide the iterates to the optimal solution of \eqref{eq:convex-saddle-point}, assuming that there is a budget constraint on the dual variables. In other words, if the dual variables could have been freely chosen from $[0,\infty)$, then the iterates of the dynamics would converge to the optimal solution of \eqref{eq:convex-saddle-point}. However, the budget constraints on the dual variables do not allow us to penalize the violated constraints arbitrarily large and force them to be feasible. Therefore, the solutions that are obtained from state-network updates may not necessarily generate a feasible solution to \eqref{eq:convex-saddle-point}. Nevertheless, this approach allows us to view the state-network dynamics as an iterative primal-dual algorithm guided by KKT optimality conditions for solving a saddle-point problem with box constraints on the dual variables. Alternatively, the state-dependent network dynamics can be viewed as Nash dynamics in a zero-sum game between a network player and a state player, with budget constraints on the action set of the network player. In the following, we use these observations to develop Lyapunov stable and convergent state-dependent network dynamics via static saddle-point problems.

\begin{theorem}\label{thm:subgradient}
Let $g_{ij}(x_i,x_j)\in \mathcal{C}$ be a symmetric convex function and $f_i(x_i)\in \mathcal{C}$ be a convex function. Assume that agent $i\in[n]$ updates its state as
\begin{align}\label{eq:dynamics-discrete}
x_i^{k+1}=x_i^k-\alpha^k\Big(\frac{\partial}{\partial x_i}f_i(x_i^k)+\!\!\!\!\sum_{j\in N_i(\boldsymbol{x}^k)}\!\!\!\frac{\partial}{\partial x_i}g_{ij}(x_i^k,x_j^k)\Big),
\end{align}   
where $N_i(\boldsymbol{x}^k)=\{j: g_{ij}(x_i^k,x_j^k)>0\}$ denotes the set of neighbors of agent $i$ at time $k$. Then, for any positive sequence $\alpha^k=\gamma_k[\sum_i\big(\frac{\partial f_i(x_i^k)}{\partial x_i}+\sum_{j\in N_i(\boldsymbol{x}^k)}\!\frac{\partial g_{ij}(x_i^k,x_j^k)}{\partial x_i}\big)^2]^{-\frac{1}{2}}$, with $\lim_k\gamma_k=0$ and $\sum_{k}\gamma_k=\infty$, the dynamics \eqref{eq:dynamics-discrete} converge to an equilibrium $\boldsymbol{x}^*$. Moreover, for sufficiently small $\alpha^k$, $V(\boldsymbol{x})=\|\boldsymbol{x}-\boldsymbol{x}^*\|^2$ serves as a Lyapunov function.
\end{theorem}
\begin{proof}
Let us consider the Lagrangian function
\begin{align}\nonumber
L(\boldsymbol{x},\boldsymbol{\lambda})=f(\boldsymbol{x})+\frac{1}{2}\sum_{i,j}\lambda_{ij}g_{ij}(x_i,x_j),
\end{align} 
where $f(\boldsymbol{x})=\sum_{i=1}^nf_i(x_i)$, and suppose that we want to solve the following saddle-point problem with box constraints on the dual variables:
\begin{align}\nonumber
\min_{\boldsymbol{x}\in \mathbb{R}^n}\max_{\boldsymbol{\lambda}\in [0,1]^{n(n-1)}}L(\boldsymbol{x},\boldsymbol{\lambda})&=\min_{\boldsymbol{x}\in \mathbb{R}^n}\max_{\boldsymbol{\lambda}\in [0,1]^{n(n-1)}}\{f(\boldsymbol{x})+\frac{1}{2}\sum_{i,j}\lambda_{ij}g_{ij}(x_i,x_j)\}\cr 
&=\min_{\boldsymbol{x}\in \mathbb{R}^n}\big\{f(\boldsymbol{x})+\frac{1}{2}\sum_{i,j}\max\{g_{ij}(x_i,x_j),0\}\big\}.
\end{align} 
If we define $\Phi(\boldsymbol{x}):=f(\boldsymbol{x})+\frac{1}{2}\sum_{i,j}\max\{g_{ij}(x_i,x_j),0\}$ and note that for any $i,j$, $\max\{g_{ij}(x_i,x_j),0\}$ is a convex function, we can easily see that $\Phi(\boldsymbol{x})$ is a convex function of $\boldsymbol{x}$. Therefore, if a subgradient algorithm is applied to the unconstrained convex problem $\min_{\boldsymbol{x}\in\mathbb{R}^n}\Phi(\boldsymbol{x})$ with an appropriate choice of step sizes $\alpha^k, k=1,2,\ldots$, the generated sequence will converge to a minimizer of $\Phi(\boldsymbol{x})$ denoted by $\boldsymbol{x}^*$. More precisely, let us use $g^k$ to denote the subgradient of $\Phi(\boldsymbol{x})$ at $\boldsymbol{x}^k$. Then, it is known that the discrete time dynamics
\begin{align}\label{eq:subgradient-dynamics}
\boldsymbol{x}^{k+1}=\boldsymbol{x}^k-\alpha_kg^k,
\end{align}  
with diminishing step length $\alpha^k=\frac{\gamma_k}{\|g^k\|}$ with $\lim_k\gamma_k=0$ and $\sum_{k}\gamma_k=\infty$, will converge to $\boldsymbol{x}^*$ \cite{boyd2004convex}. Now, let $J=\{(r,s): g_{rs}(x_r^k,x_s^k)>0\}$ and $\bar{J}=\{(r,s): g_{rs}(x_r^k,x_s^k)\leq 0\}$. Then, for every $(r,s)\in J$, the function $\max\{g_{rs}(x_r,x_s),0\}$ has a unique subgradient at $\boldsymbol{x}^k$, which is $\nabla g_{rs}(x^k_r,x^k_s)$. Moreover, for every $(r,s)\in \bar{J}$, the minimum of the convex function $\max\{g_{rs}(x_r,x_s),0\}$ equals $0$ and is achieved at $\boldsymbol{x}^k$. Thus, $\boldsymbol{0}$ is a subgradient of $\max\{g_{rs}(x_r,x_s),0\}$ at $\boldsymbol{x}^k$ for every $(i,j)\in \bar{J}$. Through use of the additivity rule of the subgradient, we conclude that $g^k=\nabla f(\boldsymbol{x})+\frac{1}{2}\sum_{(r,s)\in J}\nabla g_{rs}(x^k_r,x^k_s)$ is a subgradient for $\Phi(\cdot)$ at $\boldsymbol{x}^k$. In particular, the $i$th component of $g^k$ is given by  
\begin{align}\label{eq:subgradient}
g_i^k&=\frac{\partial}{\partial x_i}f(\boldsymbol{x}^k)+\frac{1}{2}\sum_{j}\Big(\boldsymbol{1}_{\{g_{ij}(x^k_i,x^k_j)> 0\}} \frac{\partial g_{ij}(x^k_i,x^k_j)}{\partial x_i}+\boldsymbol{1}_{\{g_{ji}(x^k_j,x^k_i)> 0\}} \frac{\partial g_{ji}(x^k_j,x^k_i)}{\partial x_i}\Big)\cr 
&=\frac{\partial}{\partial x_i}f_i(x_i^k)+\sum_{j}\Big(\boldsymbol{1}_{\{g_{ij}(x^k_i,x^k_j)> 0\}}\cdot \frac{\partial g_{ij}(x^k_i,x^k_j)}{\partial x_i}\Big)\cr
&=\frac{\partial}{\partial x_i}f_i(x_i^k)+\sum_{j\in N_i(\boldsymbol{x}^k)}\frac{\partial g_{ij}(x^k_i,x^k_j)}{\partial x_i},
\end{align}
where $\boldsymbol{1}_{\{\cdot\}}$ is the indicator function. In the above relations, the second equality holds by symmetry of the functions $g_{ij}(x_i,x_j)=g_{ji}(x_j,x_i)$, and the last equality is due to the definition of an edge emergence between two nodes $i$ and $j$. By substituting \eqref{eq:subgradient} into \eqref{eq:subgradient-dynamics}, we obtain the desired dynamics \eqref{eq:dynamics-discrete}.  

Finally, we can use the definition of the subgradient to write 
\begin{align}\nonumber
\|\boldsymbol{x}^{k+1}-\boldsymbol{x}^*\|^2&=\|\boldsymbol{x}^{k}-\boldsymbol{x}^*-\alpha^kg^{k}\|^2\cr 
&=\|\boldsymbol{x}^{k}-\boldsymbol{x}^*\|^2+(\alpha^k)^2\|g^k\|^2-2\alpha^k(g^{k})^T(\boldsymbol{x}^{k}-\boldsymbol{x}^*)\cr 
&\leq \|\boldsymbol{x}^{k}-\boldsymbol{x}^*\|^2+(\alpha^k)^2\|g^k\|^2-2\alpha^k(\Phi(\boldsymbol{x}^{k})-\Phi(\boldsymbol{x}^*)).
\end{align}
Therefore, for any $\alpha^{k}\in [0, \frac{2(\Phi(\boldsymbol{x}^{k})-\Phi(\boldsymbol{x}^*))}{\|g^k\|^2}]$, we have $\|\boldsymbol{x}^{k+1}-\boldsymbol{x}^*\|^2\leq \|\boldsymbol{x}^{k}-\boldsymbol{x}^*\|^2$. 
\end{proof}  

\smallskip
\begin{remark}
Let $X$ be the set of minimizers of $\min_{\boldsymbol{x}\in \mathbb{R}^n}\Phi(\boldsymbol{x})$, which is a nonempty closed convex set. Moreover, let $d(\boldsymbol{x},X)=\|\boldsymbol{x}-\Pi_X[\boldsymbol{x}]\|$ be the minimum distance of the point $\boldsymbol{x}$ from the set $X$, where $\Pi_X[\boldsymbol{x}]$ is the projection of $\boldsymbol{x}$ on the set $X$. As for any $\boldsymbol{x}^*\in X$, $\|\boldsymbol{x}^{k+1}-\boldsymbol{x}^*\|^2\leq \|\boldsymbol{x}^{k}-\boldsymbol{x}^*\|^2$, by choosing $\boldsymbol{x}^*=\Pi_X[x^k]$ we obtain
\begin{align}\nonumber
d(\boldsymbol{x}^{k+1}, X)=\|\boldsymbol{x}^{k+1}\!-\!\Pi_X[\boldsymbol{x}^{k+1}]\|\leq \|\boldsymbol{x}^{k+1}\!-\!\Pi_X[\boldsymbol{x}^{k}]\|\leq  \|\boldsymbol{x}^{k}\!-\!\Pi_X[\boldsymbol{x}^{k}]\|=d(\boldsymbol{x}^{k}, X).
\end{align}
Thus, for a sufficiently small step size $\alpha^{k}\in [0, \frac{2(\Phi(\boldsymbol{x}^{k})-\Phi(\boldsymbol{x}^*))}{\|g^k\|^2}]$, the distance from the iterates \eqref{eq:dynamics-discrete} to the optimal set $X$ also serves as a Lyapunov function. 
\end{remark}

\smallskip
\begin{example}
An interesting special case of Theorem \ref{thm:subgradient} is when $f_i(x_i)=0, \forall i$, and the set of constraints $\{g_{ij}(x_i,x_j)\leq 0, \forall i,j\}$ is feasible. In this case the set of minimizers of $\Phi(\boldsymbol{x})=\sum_{i,j}\max\{\frac{g_{ij}(x_i,x_j)}{2},0\}$ is precisely the feasible set $\{\boldsymbol{x}\in\mathbb{R}^n: g_{ij}(x_i,x_j)\leq 0, \forall i,j\}$. In particular, the minimum value of $\Phi(\cdot)$ is zero which is obtained at any feasible point $\boldsymbol{x}^*\in \{\boldsymbol{x}\in\mathbb{R}^n: g_{ij}(x_i,x_j)\leq 0, \forall i,j\}$. Now if the norm of the gradient of each measurement function $g_{ij}$ is bounded above by a constant $G$, we can write,
\begin{align}\label{eq:epsilon-feasible}
\frac{2(\Phi(\boldsymbol{x}^{k})-\Phi(\boldsymbol{x}^*))}{\|g^k\|^2}&=\frac{2\Phi(\boldsymbol{x}^{k})}{\|g^k\|^2}=\frac{\sum_i\sum_{j\in N_i(\boldsymbol{x}^k)}g_{ij}(x^k_i,x^k_j)}{\sum_i\big(\sum_{j\in N_i(\boldsymbol{x}^k)}\frac{\partial}{\partial x_i}g_{ij}(x^k_i,x^k_j)\big)^2}\cr 
&\ge \frac{\sum_i\sum_{j\in N_i(\boldsymbol{x}^k)}g_{ij}(x^k_i,x^k_j)}{n\sum_{i,j}\big(\frac{\partial}{\partial x_i}g_{ij}(x^k_i,x^k_j)\big)^2}\ge \frac{\max_{i,j}\{g_{ij}(x^k_i,x^k_j)\}}{n^2G^2}. 
\end{align} 
Let us define the $\epsilon$-equilibrium set as the set of all the points for which each constraint $g_{ij}(x_i,x_j)\leq 0$ is violated by at most $\epsilon$, i.e., $X_{\epsilon}=\{\boldsymbol{x}\in\mathbb{R}^n: g_{ij}(x_i,x_j)\leq \epsilon, \forall i,j\}$, and consider the dynamics \eqref{eq:dynamics-discrete} with the constant step size $\alpha_k=\frac{\epsilon}{n^2G^2}$. Then, if $\boldsymbol{x}^k\notin X_{\epsilon}$, we have $\max_{i,j}\{g_{ij}(x^k_i,x^k_j)\}>\epsilon$, which, in view of \eqref{eq:epsilon-feasible}, implies that $\alpha^k\in [0, \frac{2(\Phi(\boldsymbol{x}^{k})-\Phi(\boldsymbol{x}^*))}{\|g^k\|^2}]$. Thus, as long as $\boldsymbol{x}^k\notin X_{\epsilon}$, $d(\boldsymbol{x}, X)$ serves as a Lyapunov function, and we can write
\begin{align}\nonumber
d(\boldsymbol{x}^{k+1}, X)&\leq d(\boldsymbol{x}^{k}, X)+(\alpha^k)^2\|g^k\|^2-2\alpha^k\Phi(\boldsymbol{x}^{k})\cr 
&=d(\boldsymbol{x}^{k}, X)+\frac{\epsilon^2}{n^4G^4}\|g^k\|^2-\frac{2\epsilon}{n^2G^2}\Phi(\boldsymbol{x}^{k})\cr 
&\leq d(\boldsymbol{x}^{k}, X)+\frac{\epsilon^2}{n^4G^4}n^2G^2-\frac{2\epsilon}{n^2G^2}\epsilon=d(\boldsymbol{x}^{k}, X)-\frac{\epsilon^2}{n^2G^2},
\end{align} 
where in the last inequality we have used the fact that $\|g^k\|^2\leq n^2G^2$ and $\Phi(\boldsymbol{x}^k)\ge \epsilon$ (as $\boldsymbol{x}^k\notin X_{\epsilon}$). Since $d(\boldsymbol{x}^{k},X)\ge 0, \forall k$, we conclude that after at most $\frac{d(\boldsymbol{x}^0,X)n^2G^2}{\epsilon^2}$ iterations, the state-dependent network dynamics $x_i^{k+1}=x_i^k-\frac{\epsilon}{n^2G^2}\sum_{j\in N_i(\boldsymbol{x}^k)}\!\frac{\partial}{\partial x_i}g_{ij}(x_i^k,x_j^k)$ will reach an $\epsilon$-neighborhood of its equilibrium set $X_{\epsilon}$.    
\end{example}

\subsection{Saddle-Point Dynamics with Heterogeneous Step Size}

The subgradient method is not the only algorithm for minimizing a convex function, and one can consider other algorithms that can result in different state-dependent network dynamics. The following theorem provides another multiagent network dynamics motivated by the fact that different agents often have different scaling parameters in their update rules. These dynamics can be viewed as the \emph{quasi-Newton} method \cite{boyd2004convex} in the context of multiagent network dynamics.

\begin{definition}
A function $V:\mathbb{R}^{n}\to \mathbb{R}$ is called a \emph{semi-Lyapunov function} for the discrete-time dynamics $\boldsymbol{z}^{k+1}=h(\boldsymbol{z}^k), \ k=0,1,2,\ldots$, if $V(\boldsymbol{z}^{k+1})< V(\boldsymbol{z}^k)$ for any $z^k\in \mathbb{R}^m\setminus D$, where $D$ is a measure-zero subset of $\mathbb{R}^{n}$.
\end{definition}

\begin{theorem}\label{thm:Newton}
Let $f\in\mathcal{C}$ be a convex function and $g_{ij}\in\mathcal{C}^2$ be a symmetric convex function whose zeros form a measure-zero subset $D$.  Let $\Phi(\boldsymbol{x})\!=\!f(\boldsymbol{x})+\frac{1}{2}\sum_{i,j}\max\{g_{ij}(x_i,x_j),0\}$, whose level set and subgradient at point $\boldsymbol{x}$ are given by $L_{\boldsymbol{x}}=\{\boldsymbol{y}:\Phi(\boldsymbol{y})\leq \Phi(\boldsymbol{x})\}$ and $g^{\boldsymbol{x}}$, respectively. If for any $\boldsymbol{x}\notin D$, there exists a positive-definite diagonal matrix $G^{x}$ such that $\Phi(\boldsymbol{y})\leq \Phi(\boldsymbol{x})+\!(\boldsymbol{y}-\boldsymbol{x})^Tg^x\!+\!\frac{1}{2}\|\boldsymbol{y}-\boldsymbol{x}\|^2_{G^x}, \forall\boldsymbol{y}\in L_x$,  then the dynamics
\begin{align}\label{eq:newton-dynamics}
x_i^{k+1}=x_i^k-\frac{1}{G^k_{ii}}\Big(\frac{\partial}{\partial x_i}f(\boldsymbol{x}^k)+\!\!\!\!\sum_{j\in N_i(\boldsymbol{x}^k)}\!\!\!\frac{\partial}{\partial x_i}g_{ij}(x_i^k,x_j^k)\Big), \ \ \ i\in[n],
\end{align}  
admit the semi-Lyapunov function $\Phi(\boldsymbol{x})$. Moreover, if $G^k\le m I, \forall k$, and $\boldsymbol{x}^k\in D$ for at most finitely many iterates $k$, then the dynamics \eqref{eq:newton-dynamics} will converge.
\end{theorem}
\begin{proof}
Consider the convex function $\Phi(\boldsymbol{x})=f(\boldsymbol{x})+\frac{1}{2}\sum_{i,j}\max\{g_{ij}(x_i,x_j),0\}$, which is differentiable at any point except on a measure-zero subset $D:=\{\boldsymbol{x}\in \mathbb{R}^n: g_{ij}(x_i,x_j)=0 \ \mbox{for some}\ i,j\}$. As before, we know that the $i$th component of the gradient of $\Phi(\boldsymbol{x})$ at $\boldsymbol{x}^k$ (or the subgradient, if $\boldsymbol{x}^k\in D$) is given by $g_i^k=\frac{\partial}{\partial x_i}f(\boldsymbol{x}^k)+\sum_{j\in N_i(\boldsymbol{x}^k)}\frac{\partial}{\partial x_i}g_{ij}(x^k_i,x^k_j)$. From the assumption, we know that there is a positive-definite diagonal matrix $G^k$ such that
\begin{align}\nonumber
u_k(\boldsymbol{y}):=\Phi(\boldsymbol{x}^k)+(\boldsymbol{y}-\boldsymbol{x}^k)^Tg^k+\frac{1}{2}\|\boldsymbol{y}-\boldsymbol{x}^k\|^2_{G^k}
\end{align} 
forms a quadratic upper approximation for $\Phi(\boldsymbol{y}), \forall \boldsymbol{y}\in L_{\boldsymbol{x}^k}$, and $u_k(\boldsymbol{x}^k)=\Phi(\boldsymbol{x}^k)$. On the other hand, it is easy to see that 
\begin{align}\nonumber
\boldsymbol{x}^{k+1}=\boldsymbol{x}^k-(G^k)^{-1}g^k=\argmin_{\boldsymbol{y}\in \mathbb{R}^n}u_k(\boldsymbol{y}).
\end{align}
Let us consider an arbitrary $\boldsymbol{x}^k\notin D$. $g^k=\nabla \Phi(\boldsymbol{x}^k)$, and thus $-(G^k)^{-1}g^k$ is a descent direction for any positive-definite matrix $(G^k)^{-1}$. Therefore, for sufficiently small $\delta>0$, $\Phi(\boldsymbol{x}^k-\delta (G^k)^{-1}g^k)< \Phi(\boldsymbol{x}^k)$, and hence $\boldsymbol{x}^k-\delta (G^k)^{-1}g^k\in L_{\boldsymbol{x}^k}$. Thus, the line segment $\{(1-\alpha)\boldsymbol{x}^k+\alpha\boldsymbol{x}^{k+1}, \alpha\in [0,1]\}$ intersects $L_{\boldsymbol{x}^k}$ in at least two different points (for $\alpha=0$ and $\alpha=\delta$). Now, if $\boldsymbol{x}^{k+1}\notin L_{x^k}$, that line segment must intersect with the boundary of $L_{\boldsymbol{x}^k}$ at another point $\bar{\boldsymbol{x}}:=\bar{\alpha}\boldsymbol{x}^k+(1-\bar{\alpha})\boldsymbol{x}^{k+1}$, for some $\bar{\alpha}\in (0,1)$. (Note that the level set $L_{\boldsymbol{x}^k}$ is a closed convex set.) If we apply the continuity of $\Phi(\cdot)$, we conclude that $\Phi(\bar{\boldsymbol{x}})=\Phi(\boldsymbol{x}^k)$, and thus
\begin{align}\nonumber
u_k(\boldsymbol{x}^k)=\Phi(\boldsymbol{x}^k)=\Phi(\bar{\boldsymbol{x}})\leq u_k(\bar{\boldsymbol{x}})\leq \bar{\alpha} u_k(\boldsymbol{x}^k)+(1-\bar{\alpha})u_k(\boldsymbol{x}^{k+1})<u_k(\boldsymbol{x}^k),
\end{align}   
where the first inequality results from $\Phi(\boldsymbol{y})\leq u_k(\boldsymbol{y}), \forall \boldsymbol{y}\in L_{\boldsymbol{x}^k}$, and the second ineqaulity results from the convexity of $u_k(\cdot)$. This contradiction shows that $\boldsymbol{x}^{k+1}\in L_{\boldsymbol{x}^k}$, which implies $\Phi(\boldsymbol{x}^{k+1})< \Phi(\boldsymbol{x}^{k})$. Therefore, $\Phi(\cdot)$ serves as a semi-Lyapunov function for the dynamics \eqref{eq:newton-dynamics}. In particular, the drift of this Lyapunov function at $\boldsymbol{x}^k\notin D$ equals
\begin{align}\nonumber
\Phi(\boldsymbol{x}^k)-\Phi(\boldsymbol{x}^{k+1})&\ge \Phi(\boldsymbol{x}^k)-u_k(\boldsymbol{x}^{k+1})\cr 
&= \Phi(\boldsymbol{x}^k)-\Big(\Phi(\boldsymbol{x}^k)+(\boldsymbol{x}^{k+1}-\boldsymbol{x}^k)^Tg^k+\frac{1}{2}\|\boldsymbol{x}^{k+1}-\boldsymbol{x}^k\|^2_{G_k}\Big)\cr 
&=\frac{1}{2}(g^k)^T(G^k)^{-1}g^k\cr 
&=\frac{1}{2}\|g^k\|^2_{(G^k)^{-1}}.
\end{align}  
Summing the above inequality for $k=0,\ldots,K-1$, we obtain
\begin{align}\nonumber
\Phi(\boldsymbol{x}^{K})+\sum_{\{k: \boldsymbol{x}^k\in D\}}(\Phi(\boldsymbol{x}^{k+1})-\Phi(\boldsymbol{x}^{k}))\leq \Phi(\boldsymbol{x}^{0})-\frac{1}{2}\sum_{\{k: \boldsymbol{x}^k\notin D\}}\|g^{k}\|^2_{(G^k)^{-1}}.
\end{align} 
As this relation holds for any $K$, and $|\{k: \boldsymbol{x}^k\in D\}|<\infty$ by our assumption, we must have $\sum_{\{k: \boldsymbol{x}^k\notin D\}}\|g^{k}\|^2_{(G^k)^{-1}}<\infty$, and hence $\lim_{k\to \infty}\|g^{k}\|^2_{(G^k)^{-1}}=0$. Thus, if there exists $m>0$ such that $G^k\le m I, \forall k$, we get $\lim_{k\to \infty}\|g^{k}\|=0$. Since $\Phi(\cdot)$ is a convex function, we know that the set of minimizers of $\Phi(\cdot)$ is exactly the set of points that have $\boldsymbol{0}$ as their subgradient. Thus, $\{\boldsymbol{x}^k\}_{k=0}^{\infty}$ must converge to an equilibrium point that is a global minimum of the semi-Lyapunov function $\Phi(\cdot)$.
\end{proof}

A natural choice for the matrices $G^k$ in Theorem \ref{thm:Newton} is the Hessian matrix $\nabla^2 \Phi(\boldsymbol{x}^k)$, which is used in the Newton method for minimizing a smooth convex function. However, in practice, it is often easier to work with a sparse modification of $\nabla^2\Phi(\boldsymbol{x}^k)$ given by a diagonal matrix containing only the diagonal entries of $\nabla^2 \Phi(\boldsymbol{x}^k)$. In particular, to assure positive definiteness, an identity matrix is added to such a diagonal matrix to form the quasi-Newton update rule. Using such a quasi-Newton method for minimizing $\Phi(\cdot)$, one obtains the following state-dependent network dynamics:
\begin{align}\label{eq:hesian-newton}
x_i^{k+1}=x_i^k-t_k\frac{\frac{\partial}{\partial x_i}f(\boldsymbol{x}^k)+\!\sum_{j\in N_i(\boldsymbol{x}^k)}\!\frac{\partial}{\partial x_i}g_{ij}(x_i^k,x_j^k)}{1+\frac{\partial^2}{\partial x^2_i}f(\boldsymbol{x}^k)+\!\sum_{j\in N_i(\boldsymbol{x}^k)}\!\frac{\partial^2}{\partial x^2_i}g_{ij}(x_i^k,x_j^k)},
\end{align} 
where $t_k$ is an appropriately chosen step size obtained using a line search or diminishing rule. In fact, it is known that for a sufficiently small neighborhood of the minimizers of $\Phi(\cdot)$, the Newton method with step size $t_k=1$ will converge quadratically fast to the set of optimal points \cite{boyd2004convex}. Therefore, we obtain a simple explanation for the convergence properties and equilibrium points of seemingly complex state-dependent network dynamics \eqref{eq:hesian-newton} by using the well-known quasi-Newton method. In particular, this view provides a rigorous explanation of why the trajectories of the state-dependent network dynamics of the form \eqref{eq:hesian-newton} (such as the HK model) converge exponentially fast as they get close to their equilibrium points.       

\begin{example}\label{ex:saddle-hk}
Let us consider a special case in which $f=0$ and $g_{ij}(x_i,x_j)=\frac{1}{2}(x_i-x_j)^2-\frac{\epsilon_{ij}^2}{2}$, where $\epsilon_{ij}=\epsilon_{ji}>0$. Two agents $i$ and $j$ become each other's neighbors if their distance from each other is larger than $\epsilon_{ij}$. Therefore, we obtain the complement of the original HK model. Here, $\Phi(\boldsymbol{x})\!=\!\frac{1}{2}\sum_{ij}\max\{\frac{1}{2}(x_i-x_j)^2-\frac{\epsilon_{ij}^2}{2}, 0\}$, and thus for $\boldsymbol{x}^k\notin D:=\{\boldsymbol{x}: |x_i-x_j|=\epsilon_{ij}, \ \mbox{for some}\ i,j \}$, we have
\begin{align}\nonumber
&\nabla_i\Phi(\boldsymbol{x}^k)=\sum_{j\in N_i(\boldsymbol{x}^k)}(x^k_i-x^k_j)=|N_i(\boldsymbol{x}^k)|x_i^k-\!\!\!\sum_{j\in N_i(\boldsymbol{x}^k)}\!\!x^k_j,\cr
&\nabla^2_{ij}\Phi(\boldsymbol{x}^k)=\begin{cases}|N_i(\boldsymbol{x}^k)|\ \ &\mbox{if} \ i=j\\
-1\ \ &\mbox{if} \ j\in N_i(\boldsymbol{x}^k)\\
0\ \ &\mbox{else}.
\end{cases}
\end{align}
In other words, the Hessian matrix at $\boldsymbol{x}^k$ is equal to the Laplacian of the connectivity network at state $\boldsymbol{x}^k$. As a result, the quasi-Newton dynamics \eqref{eq:hesian-newton} for minimizing the piecewise quadratic function $\Phi(\boldsymbol{x})$ become
\begin{align}\nonumber
x_i^{k+1}=x_i^k-t_k\frac{|N_i(\boldsymbol{x}^k)|x_i^k-\!\sum_{j\in N_i(\boldsymbol{x}^k)}\!x^k_j}{|N_i(\boldsymbol{x}^k)|+1}.
\end{align}
In particular, for a sufficiently small choice of step size $t_k$, the function $\Phi(\boldsymbol{x})$ serves as a semi-Lyaponov function. Note that for unit step size $t_k=1$, the above dynamics can be explicitly written as
\begin{align}\label{eq:unit-quasi}
\boldsymbol{x}_i^{k+1}=\frac{\sum_{j\in N_i(\boldsymbol{x}^k)\cup\{i\}}x_j^k}{|N_i(\boldsymbol{x}^k)|+1}, \ \ i\in [n].
\end{align}
As a result, the dynamics of the complement HK model can be viewed as iterates of a quasi-Newton method with a unit step size to minimize $\Phi(\boldsymbol{x})$. Of course, for $t_k=1$, there is no reason why $\Phi(\boldsymbol{x})$ should serve as a Lyapunov function, unless the initial point of the dynamics is sufficiently close to a minimizer of $\Phi(\cdot)$ (in which case the exponentially fast convergence of the quasi-Newton method with $t_k=1$ is guaranteed). Nevertheless, the function $\Phi(\boldsymbol{x})$ is still very useful, as it globally guides the dynamics based on quasi-Newton iterates. In particular, the set of minimizers of $\Phi(\boldsymbol{x})$ characterize the equilibrium points of \eqref{eq:unit-quasi}.\footnote{In fact, using a sorted vector Lyapunov function $V(\boldsymbol{x})=sort(\{|x_i-x_j|, i\neq j\})$, one can show that the dynamics \eqref{eq:unit-quasi} do converge as $V(\boldsymbol{x})$ decreases lexicographically after each iteration.} This is because if $\lim_{k}\boldsymbol{x}^k=x^*$, we must have $\boldsymbol{x}^*=\lim_{k}\boldsymbol{x}^{k+1}=\boldsymbol{x}^*-\lim_{k}(G^k)^{-1}\nabla \Phi(\boldsymbol{x}^k)$, where here $(G^k)^{-1}=diag(\frac{1}{|N_1(\boldsymbol{x}^k)|+1},\ldots,\frac{1}{|N_n(\boldsymbol{x}^k)|+1})$. This implies that $\lim_{k\to \infty}G_k^{-1}\nabla \Phi(\boldsymbol{x}^k)=0$. As the entries of $G^{-1}_k$ are uniformly bounded below by $\frac{1}{n+1}$, we must have $\lim_{k\to \infty}\nabla \Phi(\boldsymbol{x}^k)=0$, and the result follows from the convexity of $\Phi(\cdot)$.
\end{example}

\section{Continuous-Time Constrained Saddle-Point Dynamics}\label{sec:continuous-saddle}

In this section, we extend our discrete-time saddle-point dynamics to their continuous-time counterparts and show how they can be leveraged to establish Lyapunov stability of state-dependent network dynamics. Here, because of the continuity of the time index $t\in [0,\infty)$, an edge connectivity between a pair of agents $(i,j)$ is no longer a binary event $\lambda_{ij}\in \{0,1\}$, but rather a continuous weight function of time $\lambda_{ij}(t)\in [0, 1]$. Thus, $\lambda_{ij}(t)$ can be viewed as a connectivity strength between agents $i$ and $j$ at time $t$ such that the maximum influence that the two agents can have on each other is $1$ (i.e., they are fully connected) and the minimum influence is $0$ (i.e., they are fully disconnected). 

Motivated by the method of change of variables for discrete time dynamics in Section \ref{subsec:change-variable}, we state our results for continuous-time dynamics in a more general form wherein the agents' states are transformed from $x_i$ to $p_i(x_i)$, and the network variables are transformed from $\lambda_{ij}$ to $q_{ij}(\lambda_{ij})$. Here, we assume that $p_i(\cdot),q_{ij}(\cdot)$ are continuous and nondecreasing functions such that $p_i(0)=q_{ij}(0)=0, \forall i,j$. In particular, we let the Lagrangian function have a more general form of $L(p(\boldsymbol{x}),q(\boldsymbol{\lambda}))$, as long as its partial derivatives exist, and be convex with respect to its first argument $p(\boldsymbol{x})=(p_1(x_1),\ldots,p_n(x_n))^T$ and concave with respect to its second argument $q(\boldsymbol{\lambda})=\big(q_{ij}(\lambda_{ij}), i\neq j\big)^T$.

\begin{remark}\label{rem:special-standard}
A special case of the above setting occurs when $p_i(x_i)=x_i, q_{ij}(\lambda_{ij})=\lambda_{ij}$ are identity functions, and $L(\boldsymbol{x},\boldsymbol{\lambda})=\sum_if_i(x_i)+\sum_{i\neq j}\lambda_{ij}g_{ij}(x_i,x_j), \boldsymbol{\lambda}\ge 0, \boldsymbol{x}\in \mathbb{R}^n$. It is clear that for the convex measurement functions $g_{ij},f_i$, the standard Lagrangian function $L(\boldsymbol{x},\boldsymbol{\lambda})$ is convex with respect to $\boldsymbol{x}$, and concave (i.e., linear) with respect to $\boldsymbol{\lambda}$. 
\end{remark}

To introduce a general class of continuous-time, state-dependent network dynamics, let us consider the following static constrained saddle-point problem:
\begin{align}\label{eq:saddle-point-general}
\min_{\boldsymbol{x}\in \mathbb{R}^n}\max_{\boldsymbol{\lambda}\in[0,1]^{n(n-1)}}L(p(\boldsymbol{x}),q(\boldsymbol{\lambda})).
\end{align} 
To solve the above static saddle-point problem using continuous-time dynamics, we use the idea of a \emph{gradient flow}, which was initially introduced in the seminal work of Arrow, Hurwicz, and Uzawa \cite{arrow1958studies} and subsequently used in devising primal-dual algorithms for solving constrained optimization problems \cite{feijer2010stability}. However, to adopt those dynamics for our more general setting \eqref{eq:saddle-point-general}, which has both lower- and upper-bound constraints on the dual variable $\boldsymbol{\lambda}$, we introduce the following generalized gradient flow dynamics:   
\begin{align}\label{eq:flow-constrained}
&\dot{\boldsymbol{x}}(t)=-\nabla_{p(\boldsymbol{x})}L\big(p(\boldsymbol{x}),q(\boldsymbol{\lambda})\big)\cr
&\dot{\boldsymbol{\lambda}}(t)=\big[\nabla_{q(\boldsymbol{\lambda})}L\big(p(\boldsymbol{x}),q(\boldsymbol{\lambda})\big)\big]^{[0,1]}_{\boldsymbol{\lambda}},
\end{align}
where $\nabla_{p(\boldsymbol{x})}L(p(\boldsymbol{x}),q(\boldsymbol{\lambda})):=\big(\frac{\partial L(p(\boldsymbol{x}),q(\boldsymbol{\lambda}))}{\partial p_1(x_1)},\ldots,\frac{\partial L(p(\boldsymbol{x}),q(\boldsymbol{\lambda}))}{\partial p_n(x_n)}\big)^T$ (the gradient vector $\nabla_{q(\boldsymbol{\lambda})}L\big(p(\boldsymbol{x}),q(\boldsymbol{\lambda})\big)$ is defined analogously), and $[a]^{[0,1]}_{\lambda}$ denotes the projection of the network dynamics to the unit interval, i.e.,
\begin{align}\nonumber
[a]^{[0,1]}_{\lambda}=\begin{cases}\min\{0,a\}, \ \ \ &\mbox{if} \ \lambda=1\\
a &\mbox{if} \ 0<\lambda<1\\
\max\{0,a\} &\mbox{if} \ \lambda=0.
\end{cases}
\end{align}
When $a$ is a vector rather than a scalar, the above projection is taken coordinatewise. The reason for introducing such a projection is that if for a pair of agents $(i,j)$ we have $\lambda_{ij}(t)\in(0,1)$, the edge variable $\lambda_{ij}(t)$ has not hit the boundary points $\{0,1\}$, and it can freely increase or decrease without violating the box constraint $\lambda_{ij}(t) \in[0,1]$. However, if $\lambda_{ij}(t)=1$, then this edge variable is only allowed to decrease, and thus $\dot{\lambda}_{ij}(t)\leq 0$. Therefore, if $\frac{\partial L(p(\boldsymbol{x}),q(\boldsymbol{\lambda}))}{\partial q_{ij}(\lambda_{ij})}\ge 0$, we set $\dot{\lambda}_{ij}(t)=0$ to block any further increase of $\lambda_{ij}(t)$. Similarly, if $\lambda_{ij}(t)=0$ and $\frac{\partial L(p(\boldsymbol{x}),q(\boldsymbol{\lambda}))}{\partial q_{ij}(\lambda_{ij})}\leq 0$, we set $\dot{\lambda}_{ij}(t)=0$ to block any further decrease of $\lambda_{ij}(t)$. Therefore, \eqref{eq:flow-constrained} provides a fairly general class of continuous-time, state-dependent network dynamics in which the strength of the edge connectivity changes dynamically as a function of the state variables. 

\begin{remark}
It is worth noting that in the special setting of Remark \ref{rem:special-standard}, the network dynamics in \eqref{eq:flow-constrained} decompose to a simple form of $\dot{\lambda}_{ij}(t)=\big[g_{ij}(x_i(t),x_j(t))\big]^{[0,1]}_{\lambda_{ij}}, \forall i,j$. Thus, the more distant two agents $i$ and $j$ are from each other (i.e., the larger the measurement value $g_{ij}(x_i,x_j)$), the faster the edge connectivity between them grows and until it achieves its maximum connectivity at $1$. That is analogous to the emergence of an edge between agents $i$ and $j$ if $g_{ij}(x_i,x_j)>0$ in the discrete-time setting.     
\end{remark}

To establish the Lyapunov stability of the continuous-time state-network dynamics \eqref{eq:flow-constrained}, let $(\bar{\boldsymbol{x}},\bar{\boldsymbol{\lambda}})$ be a saddle-point solution to \eqref{eq:saddle-point-general}. Note that by continuity and the convex-concave property of the Lagrangian function, the existence of a saddle-point in \eqref{eq:saddle-point-general} is always guaranteed. Let us define $P_i(x_i):=\int_{\bar{x}_i}^{x_i}p_i(s)ds$ and $Q_{ij}(\lambda_{ij}):=\int_{\bar{\lambda}_{ij}}^{\lambda_{ij}}q_{ij}(s)ds$, where we note that by continuity and the monotonicity of $p_i,q_{ij}$, the functions $P_i$ and $Q_{ij}$ are differentiable convex functions. Now we are ready to state the main result of this section.

\begin{theorem}\label{eq:thm-continuous}
Let $L(p(\boldsymbol{x}),q(\boldsymbol{\lambda}))$ be a convex function in $p(\boldsymbol{x})$ and a concave function in $q(\boldsymbol{\lambda})$. Then, the continuous-time state-dependent network dynamics \eqref{eq:flow-constrained} are Lyapunov stable. In particular,   
\begin{align}\nonumber
V(\boldsymbol{x},\boldsymbol{\lambda}):=\sum_{i=1}^{n}D_{P_i}(x_i,\bar{x}_i)+\sum_{i\neq j}D_{Q_{ij}}(\lambda_{ij},\bar{\lambda}_{ij})
\end{align}
serves as a Lyapunov function for the dynamics \eqref{eq:flow-constrained}, where $D_{\phi}(u,v)=\phi(u)-\phi(v)-\phi'(v)(u-v)$ denotes the Bregman divergence with respect to the convex function $\phi(\cdot)$. 
\end{theorem} 
\begin{proof}
Using the definition of the Bregman divergence, for every $i$ and $j$ we have
\begin{align}\nonumber
&\dot{D}_{P_i}(x_i,\bar{x}_i)=\frac{\partial D_{P_i}(x_i,\bar{x}_i) }{\partial x_i}\dot{x}_i=-\big(p_i(x_i)-p_i(\bar{x}_i)\big)\frac{\partial L(p(\boldsymbol{x}),q(\boldsymbol{\lambda}))}{\partial p_i(x_i)},\cr
&\dot{D}_{Q_{ij}}(\lambda_{ij},\bar{\lambda}_{ij})=\frac{\partial D_{Q_{ij}}(\lambda_{ij},\bar{\lambda}_{ij})}{\partial \lambda_{ij}}\dot{\lambda}_{ij}=\big(q_{ij}(\lambda_{ij})-q_{ij}(\bar{\lambda}_{ij})\big)\big[\frac{\partial L(p(\boldsymbol{x}),q(\boldsymbol{\lambda}))}{\partial q_{ij}(\lambda_{ij})}\big]_{\lambda_{ij}}^{[0,1]}.
\end{align} 
Now we can write
\begin{align}\nonumber
\dot{V}(\boldsymbol{x},\boldsymbol{\lambda})\!&=\!-\!\sum_i\!\big(p_i(x_i)\!-\!p_i(\bar{x}_i)\big)\frac{\partial L(p(\boldsymbol{x}),q(\boldsymbol{\lambda}))}{\partial p_i(x_i)}+\sum_{i\neq j}\!\big(q_{ij}(\lambda_{ij})\!-\!q_{ij}(\bar{\lambda}_{ij})\big)\big[\frac{\partial L(p(\boldsymbol{x}),q(\boldsymbol{\lambda}))}{\partial q_{ij}(\lambda_{ij})}\big]_{\lambda_{ij}}^{[0,1]}\cr 
&\leq\!-\!\sum_i\!\big(p_i(x_i)\!-\!p_i(\bar{x}_i)\big)\frac{\partial L(p(\boldsymbol{x}),q(\boldsymbol{\lambda}))}{\partial p_i(x_i)}+\sum_{i\neq j}\!\big(q_{ij}(\lambda_{ij})\!-\!q_{ij}(\bar{\lambda}_{ij})\big)\frac{\partial L(p(\boldsymbol{x}),q(\boldsymbol{\lambda}))}{\partial q_{ij}(\lambda_{ij})}\cr 
&=\big(\frac{\partial L(p(\boldsymbol{x}),q(\boldsymbol{\lambda}))}{\partial p(\boldsymbol{x})}\big)^T(p(\bar{\boldsymbol{x}})-p(\boldsymbol{x}))+\big(\frac{\partial L(p(\boldsymbol{x}),q(\boldsymbol{\lambda}))}{\partial q(\boldsymbol{x})}\big)^T(q(\bar{\boldsymbol{\lambda}})-q(\boldsymbol{\lambda}))\cr 
&\leq L(p(\bar{\boldsymbol{x}}),q(\boldsymbol{\lambda}))-L(p(\boldsymbol{x}),q(\boldsymbol{\lambda}))-\big(L(p(\boldsymbol{x}),q(\bar{\boldsymbol{\lambda}}))-L(p(\boldsymbol{x}),q(\boldsymbol{\lambda}))\big)\cr
&= \Big[L(p(\bar{\boldsymbol{x}}),q(\boldsymbol{\lambda}))-L(p(\bar{\boldsymbol{x}}),q(\bar{\boldsymbol{\lambda}}))\Big]+\Big[L(p(\bar{\boldsymbol{x}}),q(\bar{\boldsymbol{\lambda}}))-L(p(\boldsymbol{x}),q(\bar{\boldsymbol{\lambda}}))\Big]\leq 0,
\end{align}
where in the above derivations the last inequality is due to the definition of the saddle-point, and the second inequality follows from the convexity/concavity of $L(\cdot)$ with respect to its first/second argument. Finally, the first inequality is obtained by considering the following three cases:
\begin{itemize}
\item If $\lambda_{ij}=0$, then $[\frac{\partial L(p(\boldsymbol{x}),q(\boldsymbol{\lambda}))}{\partial q_{ij}(\lambda_{ij})}]_{\lambda_{ij}}^{[0,1]}=\max\{0,\frac{\partial L(p(\boldsymbol{x}),q(\boldsymbol{\lambda}))}{\partial q_{ij}(\lambda_{ij})}\}\ge \frac{\partial L(p(\boldsymbol{x}),q(\boldsymbol{\lambda}))}{\partial q_{ij}(\lambda_{ij})}$, and $q_{ij}(\lambda_{ij})-q_{ij}(\bar{\lambda}_{ij})=q_{ij}(0)-q_{ij}(\bar{\lambda}_{ij})\leq 0$. 
\item If $\lambda_{ij}\in (0,1)$, then $[\frac{\partial L(p(\boldsymbol{x}),q(\boldsymbol{\lambda}))}{\partial q_{ij}(\lambda_{ij})}]_{\lambda_{ij}}^{[0,1]}=\frac{\partial L(p(\boldsymbol{x}),q(\boldsymbol{\lambda}))}{\partial q_{ij}(\lambda_{ij})}$.
\item If $\lambda_{ij}=1$, then $[\frac{\partial L(p(\boldsymbol{x}),q(\boldsymbol{\lambda}))}{\partial q_{ij}(\lambda_{ij})}]_{\lambda_{ij}}^{[0,1]}=\min\{0,\frac{\partial L(p(\boldsymbol{x}),q(\boldsymbol{\lambda}))}{\partial q_{ij}(\lambda_{ij})}\}\leq \frac{\partial L(p(\boldsymbol{x}),q(\boldsymbol{\lambda}))}{\partial q_{ij}(\lambda_{ij})}$, and $q_{ij}(\lambda_{ij})-q_{ij}(\bar{\lambda}_{ij})=q_{ij}(1)-q_{ij}(\bar{\lambda}_{ij})\ge  0$. 
\end{itemize} 
Thus, in either of the above cases, we have 
\begin{align}\nonumber
(q_{ij}(\lambda_{ij})-q_{ij}(\bar{\lambda}_{ij}))[\frac{\partial L(p(\boldsymbol{x}),q(\boldsymbol{\lambda}))}{\partial q_{ij}(\lambda_{ij})}]_{\lambda_j}^{[0,1]}\leq (q_{ij}(\lambda_{ij})-q_{ij}(\bar{\lambda}_{ij}))\frac{\partial L(p(\boldsymbol{x}),q(\boldsymbol{\lambda}))}{\partial q_{ij}(\lambda_{ij})},
\end{align}
and the result follows. 
\end{proof} 

It is worth noting that Theorem \ref{eq:thm-continuous} is a continuous-time counterpart of Theorem \ref{thm:subgradient} in the sense that in both theorems, the Bregman distance of the iterates to a saddle-point serves as a Lyapunov function. However, because of the continuity of the network variables in the continuous-time model, the choice of the step size becomes irrelevant in Theorem \ref{eq:thm-continuous}, while for the discrete-time counterpart, the step sizes should be small enough to guarantee the convergence of the dynamics.  

\section{Simulations}\label{sec:simu}

In this section, we describe several numerical experiments relating to social science to justify our theoretical results. In the first experiment, we considered a class of averaging dynamics with negative weights motivated by the presence of antagonistic relations in social groups such as Altafini's averaging dynamics \cite{altafini2012dynamics}. The dynamics that we consider have the form of $x_i^{k+1}=\frac{|N_i(\boldsymbol{x}^k)|}{|N_i(\boldsymbol{x}^k)|+1}x_i^k-\frac{\sum_{j\in N_i(\boldsymbol{x}^k)}x^k_j}{|N_i(\boldsymbol{x}^k)|+1}, i\in [n]$, where $N_i(\boldsymbol{x}^k)=\{j\in [n]\!\setminus\!\{i\}: x^k_ix^k_j\leq 1\}$. In other words, two agents $i$ and $j$ are each other's neighbors if the product of their states is less than $1$.\footnote{Similar results can be obtained if one replaces 1 in the definition of the agents' neighborhood with another constant, such as 0, in which case agents communicate only if their states have the opposite sign.} On the left side of Figure \ref{fig-1}, we depict the evolution of these dynamics for $k=25$ iterations and $n=1000$ agents, where the agents' initial states are distributed uniformly at random in the interval $[-1, 1]$. As can be seen, the agents eventually polarized to two groups and diverged to $+\infty$ and $-\infty$. On the other hand, it is easy to see that the above dynamics can be replicated with Theorem \ref{thm:majorizing} if the symmetric measurements $g_{ij}(x_i,x_j)=x_ix_j-1$ and $f_i(x_i)=\frac{x_i^2}{2}$ are chosen. The implication is that $V(x)=\sum_{i=1}^{n}\frac{x_i^2}{2}+\sum_{i,j}\min\{x_ix_j-1,0\}$ must decrease along the trajectories of the dynamics, as is shown on the right side of Figure \ref{fig-1}. However, we note that the measurements $g_{ij}(x_i,x_j)=x_ix_j-1$ are neither convex nor bounded below by a global constant. In fact, that is the main reason why the monotonically decreasing function $V(x)$ cannot guarantee the convergence of the dynamics. Therefore, the convexity and boundedness of measurements in Theorem \ref{thm:majorizing} are necessary for the convergence of the trajectories.      

\begin{figure}[!tbp]
  \centering
  \begin{minipage}[t]{0.48\textwidth}
    \includegraphics[width=\textwidth]{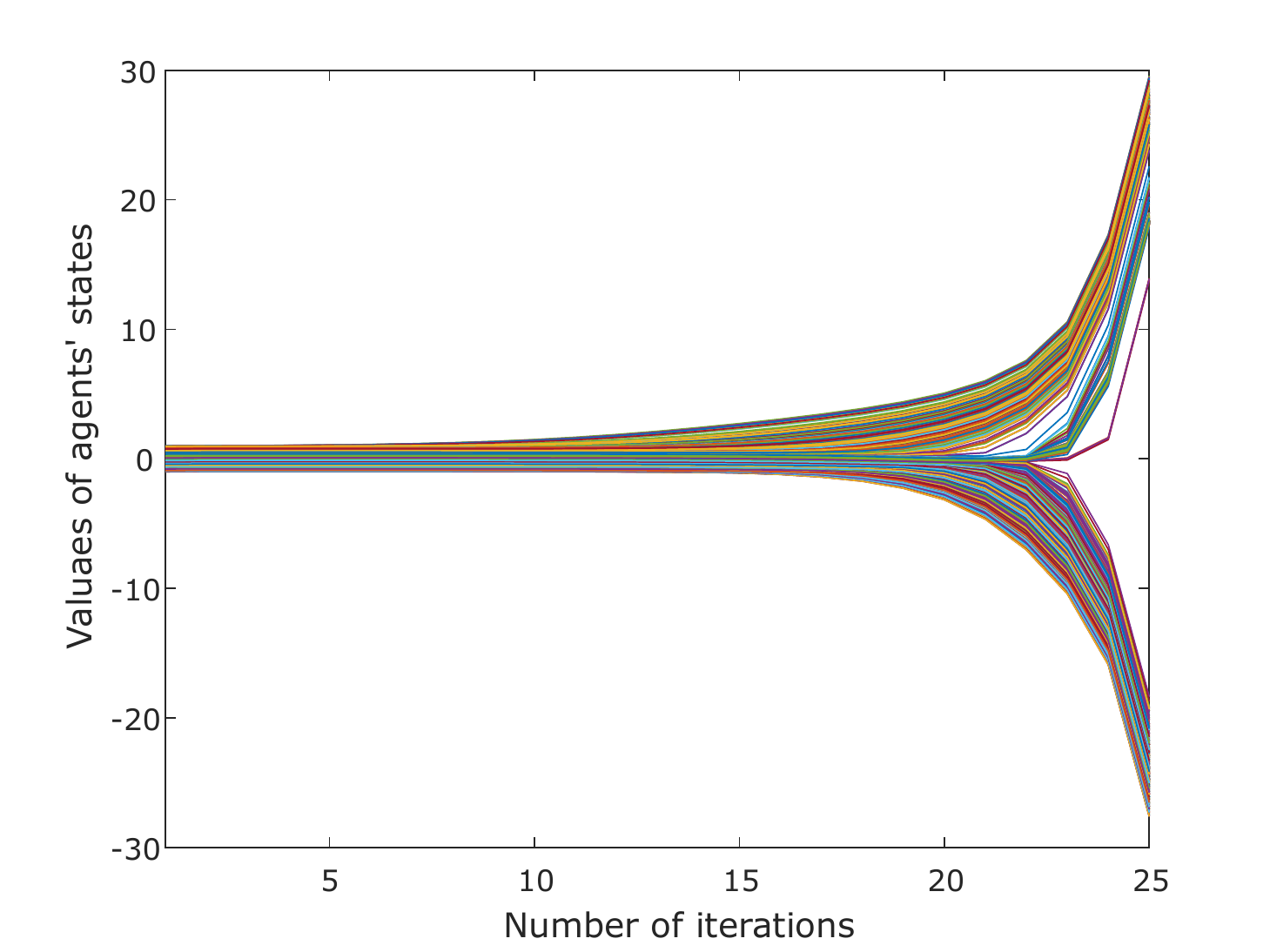}
    \hspace{-1cm}
  \end{minipage}
  \hfill
  \begin{minipage}[t]{0.48\textwidth}
    \includegraphics[width=\textwidth]{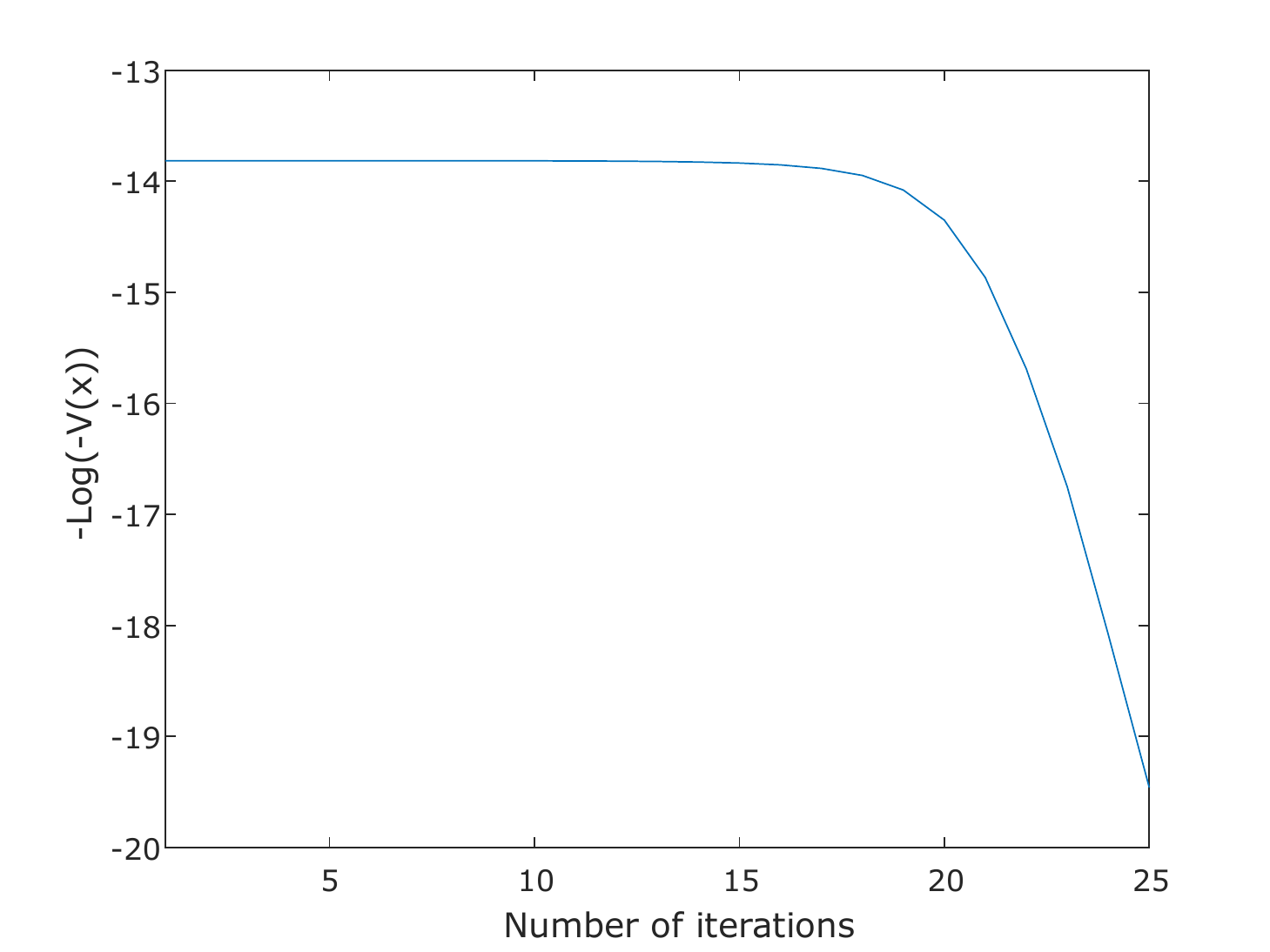}
  \end{minipage}\vspace{-0.3cm} \caption{\footnotesize{A diverging averaging dynamic with negative weights and its associated energy function.}}\label{fig-1}
\end{figure}

In the second experiment, we considered the dynamics \eqref{eq:unit-quasi} in Example \ref{ex:saddle-hk}. Here, the symmetric measurements are $g_{ij}(x_i,x_j)=\frac{(x_i-x_j)^2}{2}-\frac{20^2}{2}$ (i.e., $\epsilon_{ij}=20, \forall i\neq j$), and the set of neighbors of agent $i$ at iteration $k$ is given by $N_i(\boldsymbol{x}^k)=\{j: |x^k_i-x^k_j|\ge 20\}$. As was shown in Example \ref{ex:saddle-hk}, those dynamics are the complement of the homogeneous HK dynamics. The left side of Figure \ref{fig-2} shows the results of simulating the trajectories for $k=120$ iterations and $n=1000$ agents, with initial states distributed uniformly at random in the interval $[0, 100]$. As can be seen, although the dynamics eventually converge to several clusters, however, there is an oscillating pattern in the trajectories due to the conflicting objectives of the network structure and the agents' states. In particular, Theorem \ref{thm:Newton} suggests that the convex function $V(\boldsymbol{x})=\frac{1}{2}\sum_{i,j}\max\{\frac{(x_i-x_j)^2}{2}-\frac{20^2}{2}, 0\}$ serves as a semi-Lyapunov function that almost always decreases along the trajectories. That phenomenon is shown on the right side of Figure \ref{fig-2}, where $V(\boldsymbol{x})$ always decreases with a small jump at iteration $k=3$. In particular, the dynamics converge to a minimizer of $V(\boldsymbol{x})$.              

\begin{figure}[t]
  \centering
  \begin{minipage}[t]{0.48\textwidth}
    \includegraphics[width=\textwidth]{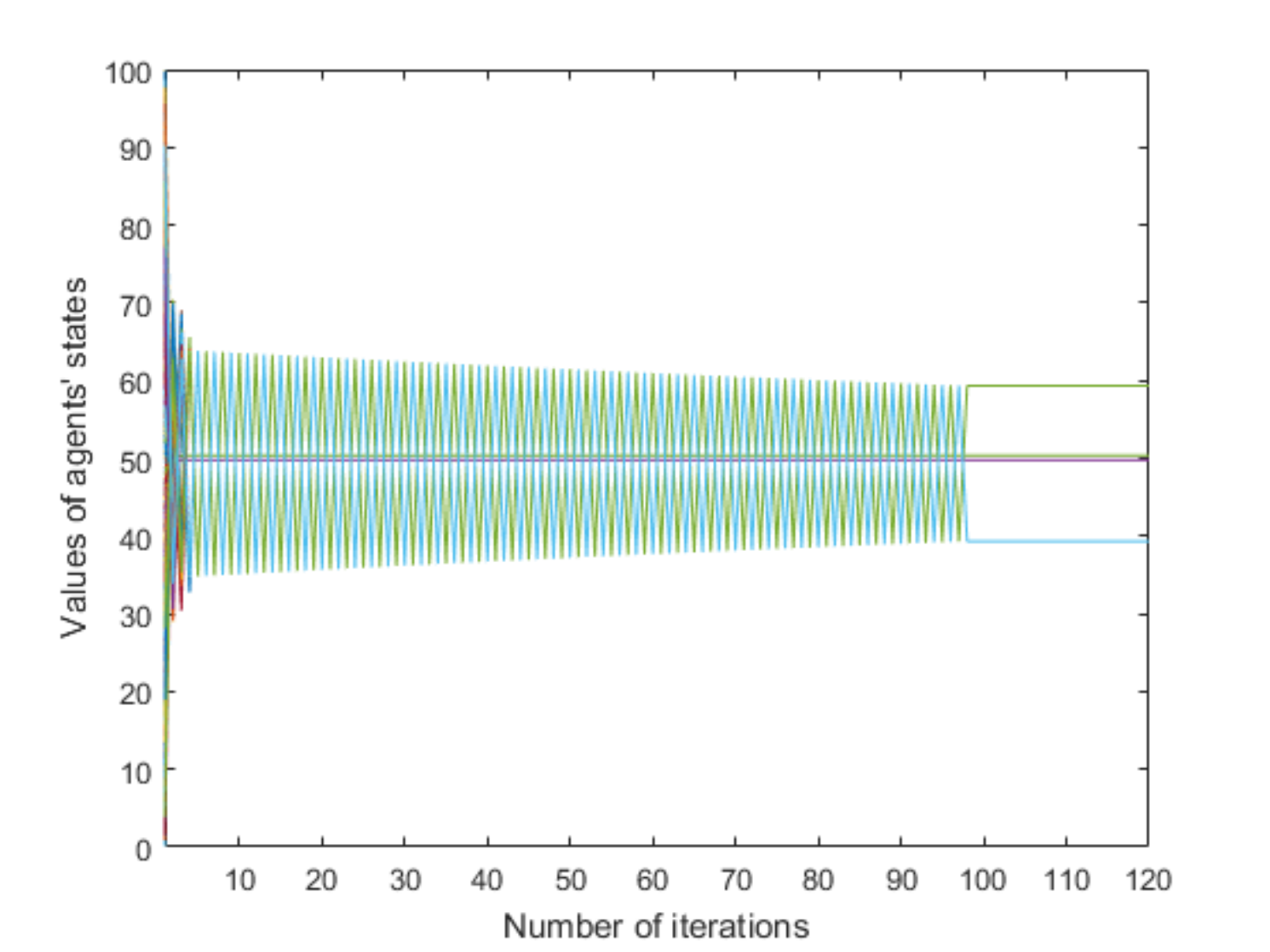}
    \vspace{-0.5cm}
  \end{minipage}
  \hfill
  \begin{minipage}[t]{0.48\textwidth}
    \includegraphics[width=\textwidth]{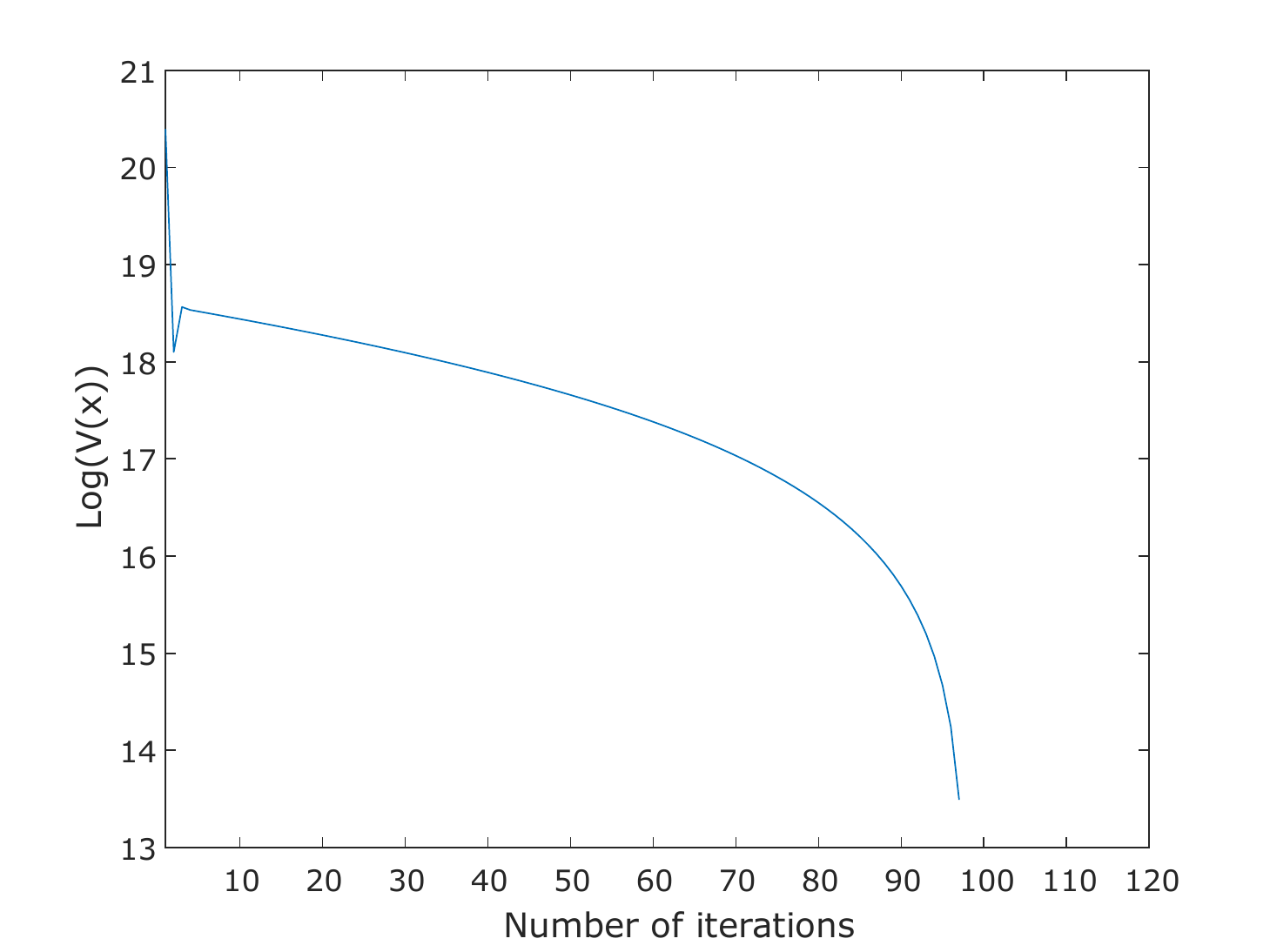}
  \end{minipage} \vspace{-0.3cm} \caption{\footnotesize{Dynamics of the complement HK model and its associated semi-Lyapunov function.}}\label{fig-2}
\end{figure}

In our last experiment, we considered \emph{heterogeneous} HK dynamics with associated dynamics $x_i^{k+1}=\frac{\sum_{j: |x^k_i-x^k_j|\leq \epsilon_i}x^k_j}{|\{j: |x^k_i-x^k_j|\leq \epsilon_i\}|}, i\in [n]$, whose Lyapunov stability and convergence have remained open for more than a decade \cite{mirtabatabaei2012opinion,chazelle2016inertial}. We can replicate the heterogeneous HK dynamics by considering the asymmetric measurements  $g_{ij}(x_i,x_j)=\frac{(x_i-x_j)^2}{2}-\frac{\epsilon^2_i}{2}$ in the statement of Theorem \ref{thm:asymmetric}. Although we do not know a choice of Bregman map $f(x)$ that can satisfy the assumption of Theorem \ref{thm:asymmetric} with respect to such asymmetric measurements, in Figure \ref{fig-3} we show that even a simple quadratic map performs quite well for the heterogeneous HK model. More precisely, based on the quadratic Bregman map $D_f(\boldsymbol{x},\boldsymbol{y})=\frac{1}{2}\|\boldsymbol{x}-\boldsymbol{y}\|^2$, where $f(\boldsymbol{x})=\sum_{i=1}^{n}\frac{x_i^2}{2}$, Theorem \ref{thm:asymmetric} suggests the Lyapunov candidate $V(\boldsymbol{x})=\sum_{i,j}\min\{\frac{(x_i-x_j)^2}{2}-\frac{\epsilon^2_i}{2},0\}$ for the heterogeneous HK dynamics. We simulated those dynamics for $k=300$ iterations over a set of $n=1000$ agents whose initial states were distributed uniformly at random in the interval $[0, 100]$. In the top two graphs in Figure \ref{fig-3}, the confidence bounds of agents $\epsilon_i, i\in [n]$ had been selected uniformly at random from the interval $[0, 10]$. In the bottom two graphs in Figure \ref{fig-3}, the confidence bounds had been generated uniformly at random from the larger interval $[0, 50]$. In both cases, one can see that the proposed function $V(\boldsymbol{x})$ performed quite well and almost always decreased along the trajectories of the heterogeneous HK. That suggests that addition of a small correction term to the proposed $V(\boldsymbol{x})$ might turn this function into a valid Lyapunov function for the heterogeneous HK dynamics.       

\begin{figure}[!tbp]
  \centering
  \begin{minipage}[t]{0.48\textwidth}
    \includegraphics[width=\textwidth]{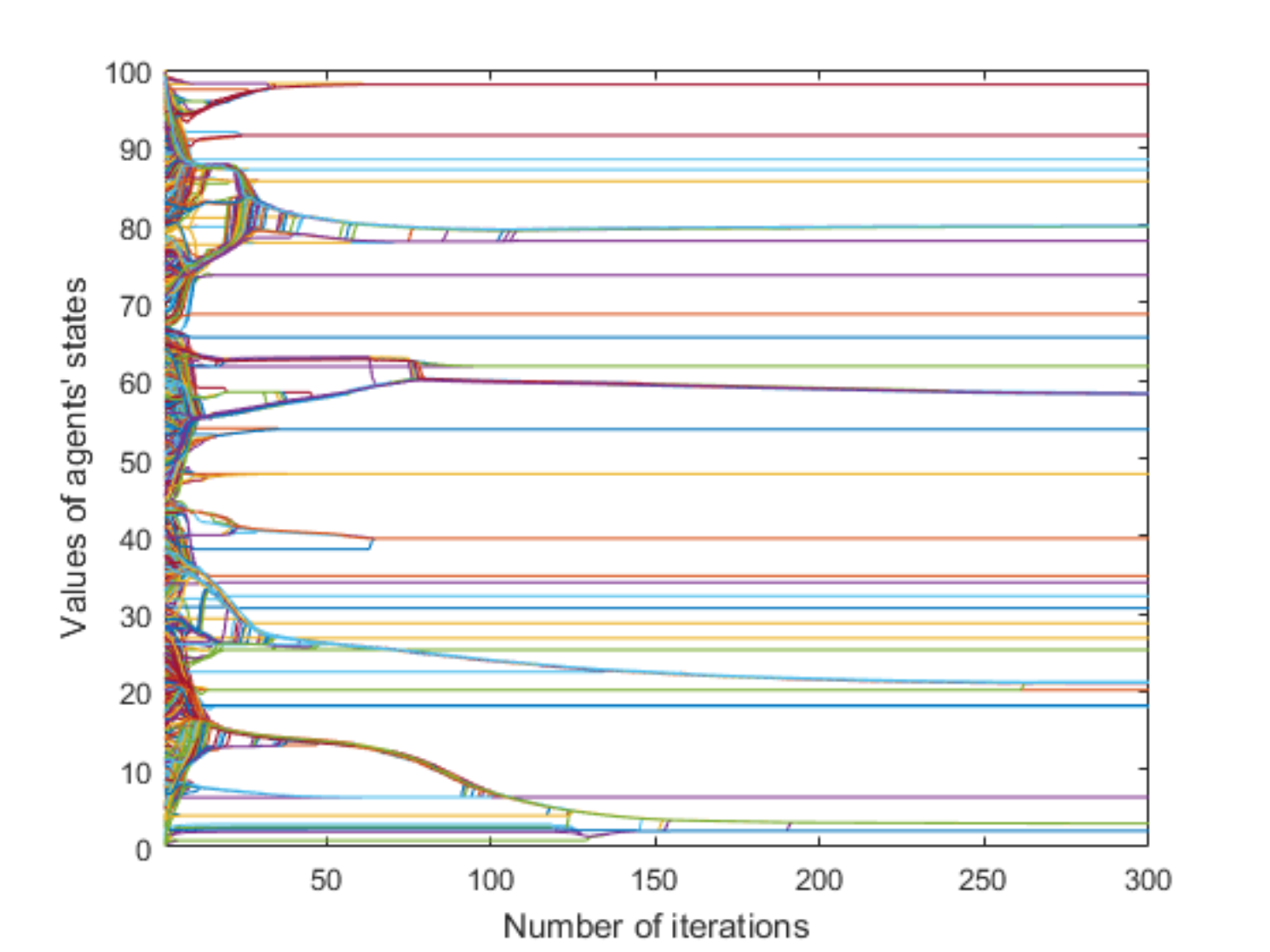}
    \vspace{-0.5cm}
  \end{minipage}
  \hfill
  \begin{minipage}[t]{0.48\textwidth}
    \includegraphics[width=\textwidth]{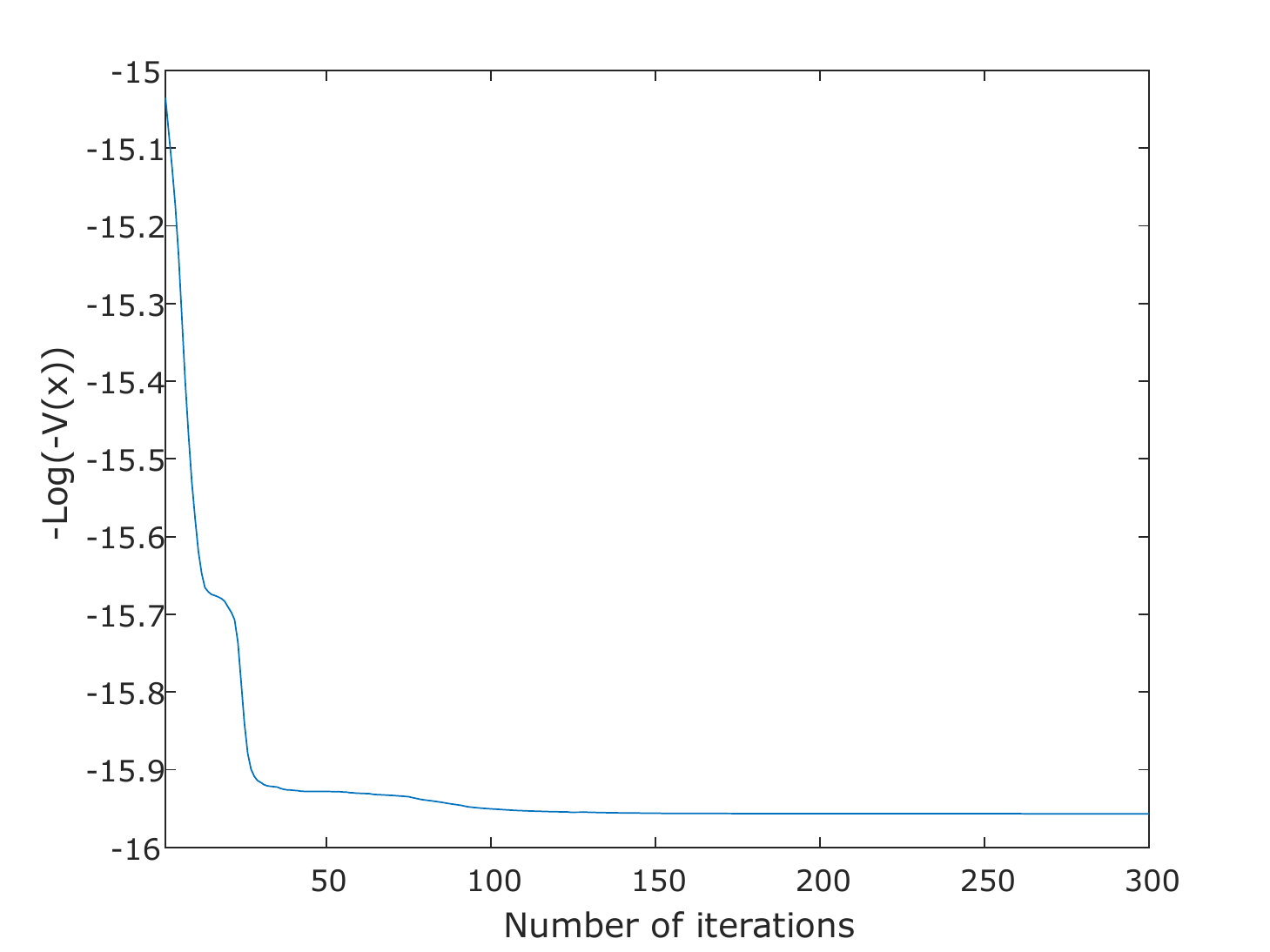}
    \vspace{-0.5cm} 
  \end{minipage}\\
  \begin{minipage}[t]{0.48\textwidth}
    \includegraphics[width=\textwidth]{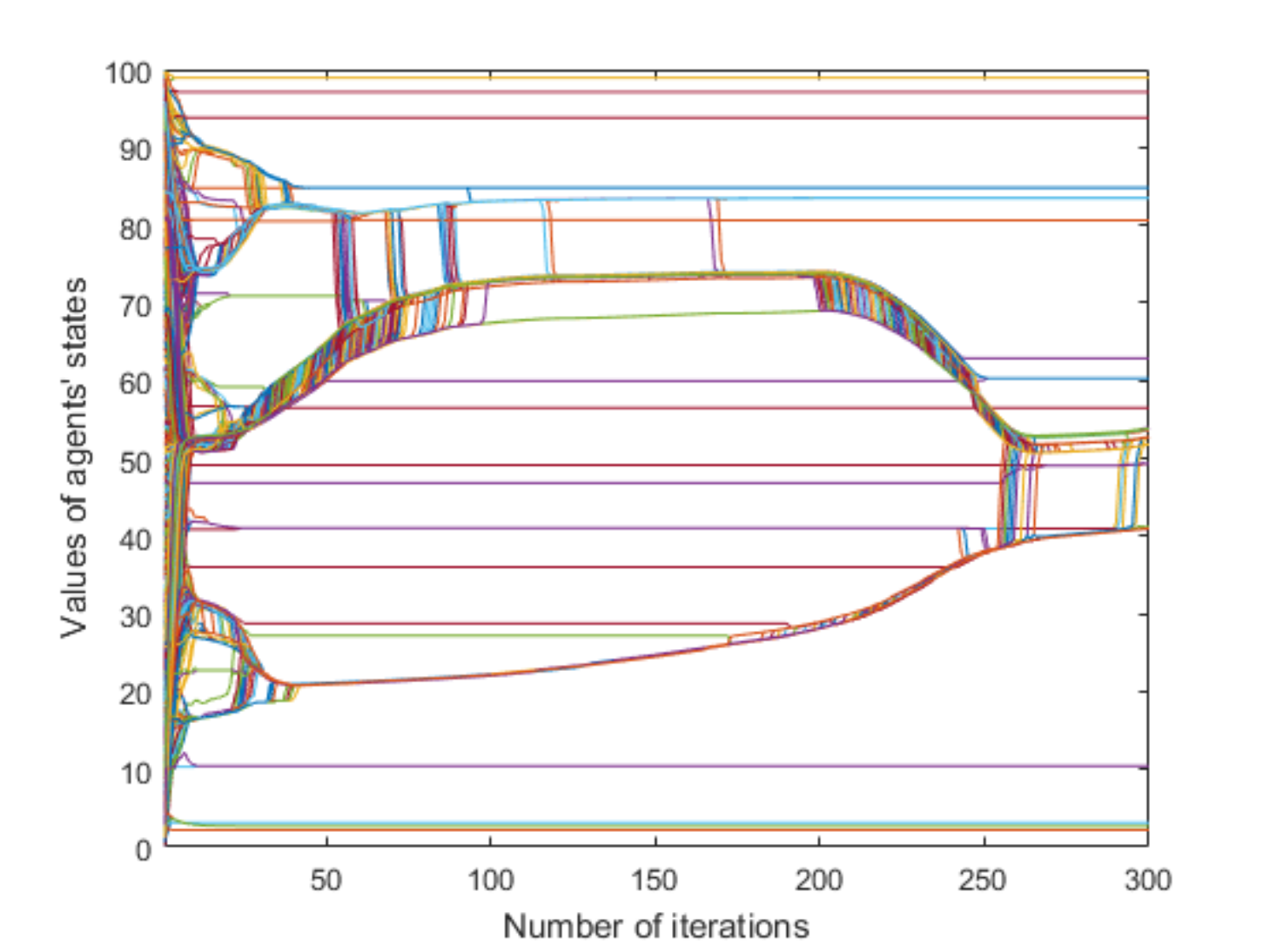}
    \vspace{-0.5cm}
  \end{minipage}
  \hfill
  \begin{minipage}[t]{0.48\textwidth}
    \includegraphics[width=\textwidth]{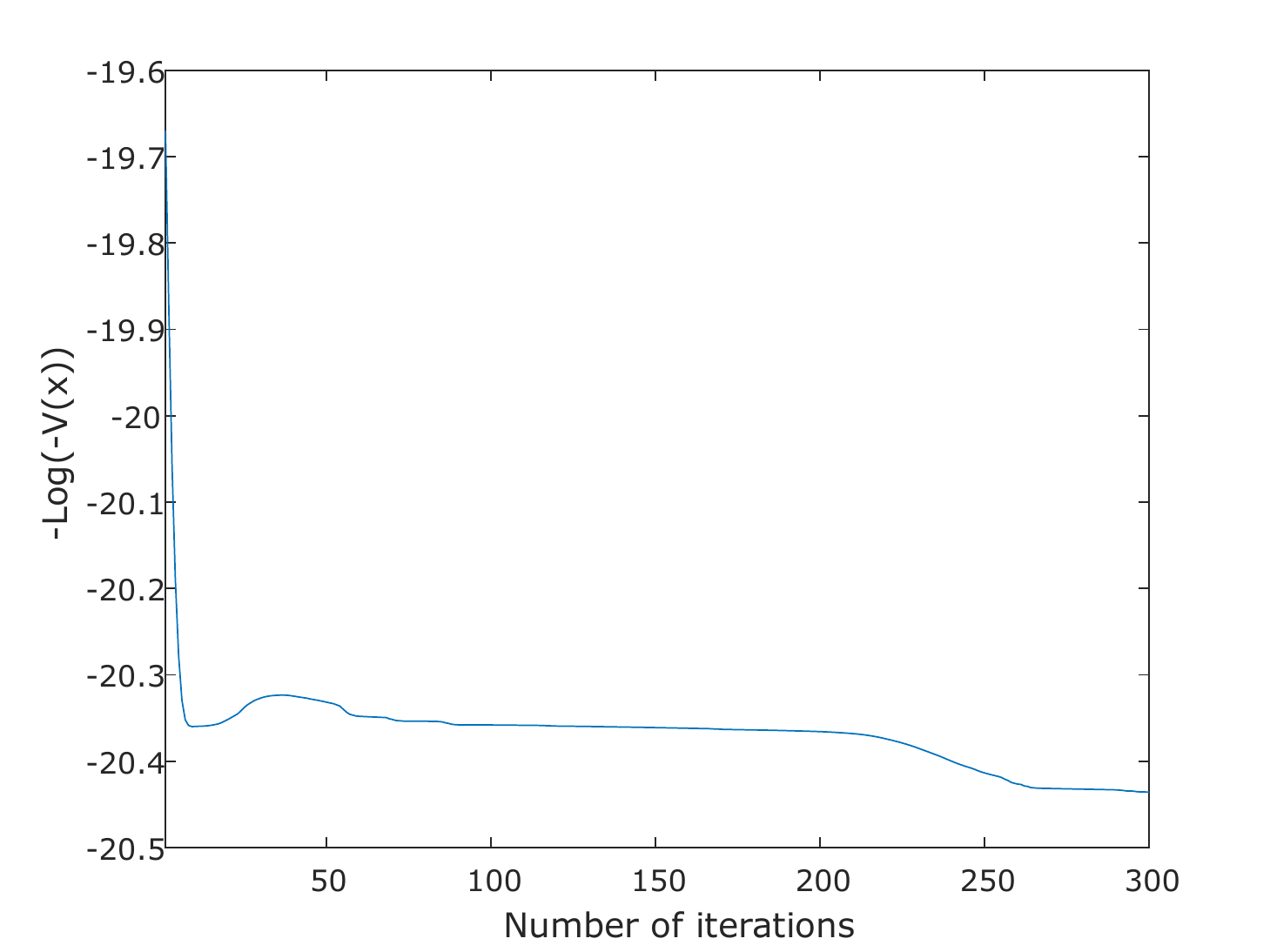}
  \end{minipage}\vspace{-0.3cm} \caption{\footnotesize{Evolution of the heterogeneous HK dynamics and their associated $V(\boldsymbol{x})$ values with uniformly generated confidence bounds from intervals $[0, 10]$ (top figures) and $[0, 50]$ (bottom figures).}}\label{fig-3}
\end{figure}

\section{Conclusions}\label{sec:conclusion}
In this paper, we developed a new framework for the stability analysis of multiagent state-dependent network dynamics. We showed that the co-evolution of the network and the state dynamics could be cast as a primal-dual optimization algorithm for a nonlinear program in which the primal updates capture the state dynamics, and the dual updates capture the network evolution. In particular, the constrained Lagrangian function serves as a Lyapunov function for the state-network dynamics. We considered our framework under two different settings: i) when the network and state dynamics are aligned, and ii) when the network and state dynamics have conflicting objectives. In the first case, we showed that the application of the BCD method with a change of variables could generate a variety of interesting state-dependent network dynamics. In particular, we provided a new technique for handling asymmetry in the network dynamics. In the second case, we reduced the stability of the state-network dynamics to a zero-sum game between the network player and the state player. This approach allowed us to establish the Lyapunov stability of multiagent systems by using saddle-point dynamics and, in particular, by using the subgradient method and the quasi-Newton method. Finally, we extended our results to a continuous-time model and provided a general class of continuous-time, state-dependent network dynamics in terms of generalized gradient flow. 

As a future direction of research, one could use \emph{augmented} Lagrangian functions or apply other optimization techniques to generate a broader class of stable state-dependent network dynamics. Moreover, in our analysis, we mainly used a quadratic upper approximation to derive the state updates. Thus, a natural extension would be to use other function approximations that include the quadratic approximation as their particular case, or to use approximations that are suitable for specific applications. 

\bibliographystyle{IEEEtran}
\bibliography{thesisrefs}

\end{document}